\documentclass[12pt,letterpaper]{article}

\usepackage{xcolor, colortbl}
\usepackage{amssymb, amsmath, amsthm}
\usepackage{graphicx}
\usepackage{natbib}
\usepackage{float}
\usepackage{setspace}
\usepackage{algorithm}
\usepackage{algorithmic}

\newtheorem{proposition}{Proposition}

\newtheorem{assumption}{Assumption}
\newtheorem{example}{Example}
\definecolor{lightgray}{gray}{0.9}

\newcommand{\ba}{\mathbf{a}}
\newcommand{\bA}{\mathbf{A}}

\newcommand{\bB}{\mathbf{B}}

\newcommand{\bD}{\mathbf{D}}
\newcommand{\be}{\mathbf{e}}
\newcommand{\bE}{\mathbf{E}}

\newcommand{\bL}{\mathbf{L}}

\newcommand{\bO}{\mathbf{O}}
\newcommand{\bP}{\mathbf{P}}
\newcommand{\bQ}{\mathbf{Q}}
\newcommand{\bR}{\mathbf{R}}

\newcommand{\bT}{\mathbf{T}}
\newcommand{\bu}{\mathbf{u}}

\newcommand{\bv}{\mathbf{v}}

\newcommand{\by}{\mathbf{y}}

\newcommand{\bZ}{\mathbf{Z}}

\renewcommand{\epsilon}{\varepsilon}

\renewcommand{\tilde}{\widetilde}
\renewcommand{\leq}{\leqslant}
\renewcommand{\geq}{\geqslant}

\newcommand{\distn}[1]{\mathcal{#1}}

\newcommand{\Em}{\mathbb E}

\newcommand{\gvn}{\,|\,}

\newcommand{\vect}[1]{\boldsymbol #1}

\newcommand{\vLambda}{\vect{\Lambda}}

\newcommand{\vSigma}{\vect{\Sigma}}

\bibpunct{(}{)}{;}{a}{,}{,}

\begin{document}

\title{Large Structural VARs with Multiple Sign and Ranking Restrictions}

\author{Joshua C. C. Chan\\
	{\small Purdue University} \\
	\and Christian Matthes \\
	{\small Indiana University} \\
	\and Xuewen Yu\\
	{\small Fudan University}
}

\date{March 2025}

\maketitle

\onehalfspacing

\thispagestyle{empty}

\begin{abstract}

\noindent Large VARs are increasingly used in structural analysis as a unified framework to study the impacts of multiple structural shocks simultaneously. However, the concurrent identification of multiple shocks using sign and ranking restrictions poses significant practical challenges to the point where existing algorithms cannot be used with such large VARs. To address this, we introduce a new numerically efficient algorithm that facilitates the estimation of impulse responses and related measures in large structural VARs identified with a large number of structural restrictions on impulse responses. The methodology is illustrated using a 35-variable VAR with over 100 sign and ranking restrictions to identify 8 structural shocks.

\bigskip


\bigskip


\end{abstract}

\newpage

\section{Introduction}

Vector autoregressions (VARs) are a workhorse model in macroeconomic forecasting and structural analysis. Among the many methodological advances in the structural VAR (SVAR) literature since the pioneering work by \citet{sims80}, two recent developments are the most prominent. First, there is a growing recognition of the need to exploit more information in structural analyses, motivated by the concern that informational deficiency (using an information set that is too small relative to that of economic agents) substantially distorts estimates of impulse responses and related objects \citep{HS91,LR93, LR94}. Starting from the seminal paper by \citet{LSZ96} that develops various medium-sized structural VARs to study the effects of monetary policy, large VARs with dozens of endogenous variables are increasingly being used in applications. This trend gained momentum after the influential work by \citet*{BGR10}, who demonstrate the benefits of including a large number of variables for both forecasting and structural analysis. Notable applications using large VARs include \citet*{CKM09}, \citet{koop13}, \citet{ER17} and \citet{Crumpetal21}.

The second development relates to the methods for identifying structural shocks. More specifically, there has been a gradual departure from conventional recursive or zero restrictions to alternative structural restrictions that are deemed to be more credible. An important class of identifying restrictions imposes sign restrictions motivated by economic theory, developed in a series of papers by \citet{Faust98}, \citet{CD02} and \citet{Uhlig05}. Extensions of this identification approach, such as ranking restrictions proposed in \citet{AD21}, are also widely used in empirical work.

The convergence of these two developments naturally requires the estimation of large structural VARs identified by imposing sign and ranking restrictions on the impulse responses. However, this remains practically infeasible in high-dimensional settings. For instance, using the popular accept-reject algorithm of \citet{RWZ10} to impose sign restrictions might take days in larger-scale applications. 
Thus, this computational burden severely limits the use of these more credible restrictions in large systems. % This computational burden thus severely limits the scope of empirical studies that aim to disentangle the impacts of multiple structural shocks using more credible restrictions.

We develop a new approach to estimate large SVARs identified using a large number of structural restrictions on impulse responses, which was until now computationally infeasible. In particular, it is applicable to situations where there are far more structural restrictions than identified shocks. The new algorithm builds upon the accept-reject algorithm of \citet{RWZ10}, which we now briefly describe to provide some perspective. First, given a uniformly drawn orthogonal matrix (i.e., a matrix drawn according to the Haar measure), \citet{RWZ10} check if the implied impulse responses satisfy all restrictions. If all the restrictions are satisfied (the draw is admissible), accept the draw and the implied impulse responses; otherwise, obtain another uniform draw and repeat the procedure. The main computational bottleneck of this algorithm comes from the fact that in high-dimensional settings with a large number of structural restrictions, it is highly unlikely that any orthogonal matrix drawn uniformly is admissible. Consequently, one typically needs to sample a huge number of orthogonal matrices to obtain one that is admissible. 

The key idea of our proposed algorithm comes from the recognition that, given a uniformly distributed orthogonal matrix, a vast collection of uniform draws can be constructed by permuting its columns and switching the signs of the columns.\footnote{Since the Haar measure is invariant under permutations and sign switches, any member of this collection is also uniformly distributed in the orthogonal group.} More importantly, all these obtained orthogonal matrices are equivalent, in the sense that they represent exactly the same structural shocks of the original orthogonal matrix, after relabeling the shocks and proper sign normalizations. Additionally, one can effectively search through this collection to locate any members that satisfy all structural restrictions with trivial computations. In this way, the new algorithm significantly increases the probability of obtaining an admissible draw with virtually no additional costs. In our benchmark setting, we impose that any identification restriction is only imposed on impact to allow for fast checking of identification restrictions. Economic theory generally only produces robust restrictions \emph{across theoretical models} only on impact, giving a justification for this approach. However, we also discuss how to extend our approach to sign restrictions at longer horizons as well as ranking restrictions along the lines of \cite{DGK14} and \cite{AD21}. 

To illustrate our proposed algorithm, we consider four applications, three empirical applications based on US data and one set of Monte Carlo simulations. First, we estimate a 15-variable VAR with more than 40 sign and ranking restrictions to identify 5 structural shocks based on empirical applications in \citet{FRS19} and \citet{chan22}. As a benchmark, we use the algorithm of \citet{RWZ10} to uniformly draw orthogonal matrices from the admissible set and compute implied impulse responses. It takes about 3.6 billion orthogonal matrices to obtain 1,000 admissible draws, and the estimation takes about 6 days on a standard desktop. In contrast, the new algorithm requires only about 31,000 orthogonal matrices to obtain 1,000 admissible draws, and the entire exercise takes about 16 seconds. We also confirm empirically that both algorithms give identical impulse responses. Second, we demonstrate how the proposed algorithm can be applied in settings with dynamic sign restrictions by replicating the classic application in \citet{Uhlig05}. We compare the proposed algorithm to those of \citet{RWZ10} and \citet{Read22}, and show that while the three algorithms produce the same impulse responses, the proposed algorithm is substantially faster. Third, we conduct a series of Monte Carlo simulations to illustrate the empirical performance of the proposed method, and show that it works well even in settings with large numbers of variables and structural shocks. 

Our third empirical application considers a larger 35-variable VAR with over 100 sign and ranking restrictions to identify 8 structural shocks: demand, investment, financial, monetary policy, government spending, technology, labor supply and wage bargaining. These macroeconomic and financial  variables are broadly similar to those of \citet{Crumpetal21} and are closely monitored by policy institutions and market participants. Our high-dimensional model provides a unified framework to study the impacts of multiple structural shocks simultaneously. In particular, this framework allows us to disentangle the impacts of different types of demand and supply shocks on key macroeconomic variables.  Even for such a large system, the estimation takes only 14 minutes. Therefore, this application demonstrates that it is practical to study the impacts of multiple structural shocks jointly in a large system using the proposed approach. % The proposed algorithm therefore allows empirical researchers... *** highlight a few results if they are interesting? *** 

Our paper contributes to the emerging literature on efficient methods for conducting structural analysis using large VARs. As noted in \citet{Crumpetal21}, central banks and policy institutions routinely monitor and forecast dozens of key macroeconomic variables, and VARs provide a convenient framework for studying the joint impacts of multiple structural shocks. To reduce the computational burden of performing structural analysis in large systems, some recent papers, such as \citet{korobilis22} and \citet{CEY22}, propose using a factor model for the reduced-form VAR errors and structural identification restrictions are placed on factor loadings. In contrast, our paper uses a standard VAR framework where structural shocks are related to the reduced-form errors through an impact matrix. Therefore, the proposed algorithm is directly applicable to a wide variety of VARs currently used for structural analysis. 

This paper also relates to the literature on efficient posterior sampling in structural VARs with informative priors on impulse responses \citep[see, e.g.,][]{Kocikecki10, BH15, BH18}. In particular, for VARs identified using sign restrictions, the proposed algorithm can be used in the first stage to generate proposal draws for an importance sampler to explore the posterior distribution that incorporates prior information on impulse responses; a recent example of such an importance sampler is given in \citet{BP23}. The proposed algorithm can thus boost the efficiency of the second-stage importance sampler and make it applicable beyond medium-sized models.

% Since a large number of proposal draws are typically required to guarantee sufficient efficiency of the importance sampler, 

% Can say something about how the proposed algorithm can be used in ohter settings, such as when posterior draws of the reduced-form model are very computationally intensive to obtain such as in \citet{Braun23}?

The remainder of this paper is organized as follows. Section \ref{s:id} first outlines the identification of shocks in a structural VAR using sign restrictions. We then introduce the proposed algorithm for generating uniform draws of the impact matrix that satisfy all the sign restrictions at impact. Finally, we discuss how the proposed algorithm can be extended to handle other commonly-used identification schemes. Section~\ref{s:comparison} considers an illustration using a 15-Variable VAR with sign restrictions to identify 5 structural shocks. We compare the speed of the proposed algorithm as well as the impulse response estimates with those obtained using the algorithm of \citet{RWZ10}. We further illustrate how the proposed algorithm can be applied in settings with dynamic sign restrictions by replicating the application in \citet{Uhlig05}. Section~\ref{s:application} considers an application that involves 35 US macroeconomic and financial variables. We use over 100 sign and ranking restrictions to identify 8 structural shocks. Lastly, Section~\ref{s:conclusion} concludes and outlines some future research directions.

% first introduces a general framework for modeling matrix-valued time-series with a flexible error covariance structure. It then offers a few different interpretations of the matrix autoregression and discusses some identification issues. Lastly, the section develops Bayesian shrinkage priors that generalize the Minnesota priors to the MAR setting. Section~\ref{s:estimation} proposes a unified approach to estimate these flexible MARs using Markov chain Monte Carlo (MCMC) methods.

 % It is an accept-reject algorithm, and often requires tens of millions of MCMC draws to get enough admissible draws. The problem is compounded when getting MCMC draws is time-consuming (in larger systems, more flexible models). 

% Many other applications also naturally call for the inclusion of a large number of time-series, such as those involving data in multiple frequencies \citep{SS15, MOS21}, disaggregated data \citep{GLMO14, ER17}, firm-level data \citep{DDLY18} and financial data \citep{CKM09}. Finally, the wide availability of large time-series datasets \citep{MN16,MN20,BLS22,BJKO23} further propels this development.

\section{Identification of Structural Shocks} \label{s:id}

In this section, we first outline the identification of structural shocks in a structural VAR using sign restrictions. In Section \ref{ss:new_algorithm} we then introduce the proposed algorithm to efficiently generate draws of the impact matrix that satisfy all the sign restrictions at impact. Section \ref{ss:extensions} further discusses how the proposed algorithm can be extended to handle other commonly used identification schemes, such as ranking restrictions.

To set the stage, let $\by_t= (y_{1,t},\ldots,y_{n,t})'$ be an $n\times 1$ vector of endogenous variables that is observed over the periods $t=1,\ldots, T.$ Consider the following VAR with $p$ lags:
\begin{align} 
	\by_t & = \ba_0 + \bA_1 \by_{t-1} + \cdots + \bA_p\by_{t-p} + \bu_t, \label{eq:VAR} \\
	\bu_t & = \bB_0\bv_t, \quad \bv_t\sim\distn{N}(\mathbf{0},\mathbf{I}_n),  \label{eq:shock} 
\end{align}
where the vector of structural shocks $\bv_t$ is related to the reduced-form errors $\bu_t$ via the impact matrix $\bB_0$ that is assumed to be non-singular. It follows that the covariance matrix of $\bu_t$ is $\vSigma\equiv \bB_0 \bB_0'$. 

One main goal of estimating the VAR in \eqref{eq:VAR}-\eqref{eq:shock} is to study the impact of structural shock $v_{j,t}$ on the endogenous variable $y_{i,t}$, $i=1,\ldots, n$ and $j=1,\ldots, m$. Specifically, the impulse response at horizon $h$ is defined to be the expected change in the conditional mean of $y_{i,t+h}$ from the $j$-th structural shock $v_{j,t}$:
\begin{equation}\label{eq:IR}
	f_{i,j,h} = \Em\left[y_{i,t+h} \gvn \bv_t = \be_j; \bB_0, \bA\right] - 
	\Em\left[y_{i,t+h}\gvn \bv_t = \mathbf{0}; \bB_0, \bA\right],
\end{equation}
where $\be_j$ is the $j$-th column of the $n$-dimensional identity matrix $\mathbf{I}_n$ and $\bA = (\ba_0, \bA_1, \ldots, \bA_p)'$ is the $k\times n$ matrix of VAR coefficients with $k=np+1$. Note that each impulse response $f_{i,j,h}$ depends implicitly on the impact matrix $\bB_0$ and the VAR coefficients $\bA$.

It is well known that under the setup in \eqref{eq:VAR}-\eqref{eq:shock}, $\bB_0$ is not point-identified: since given any orthogonal matrix $\bQ\in\bO(n)$ and $\tilde{\bB}_0 = \bB_0\bQ$, we have $\tilde{\bB}_0\tilde{\bB}_0' = \vSigma$. In other words, there is a range of impulse responses of variables to structural shocks, even if we fix the identifiable model parameters 
$(\bA,\vSigma)$. One often proceeds by restricting the set of impulse responses---e.g., by imposing economically meaningful restrictions on the impulse responses. Starting from the influential papers by \citet{Faust98}, \citet{CD02} and \citet{Uhlig05}, one prominent approach is to impose sign restrictions motivated by economic theory on the impulse responses. 

More specifically, let $s_{i,j,h}\in\{-1,0,1\}$. Then, a sign restriction on the impulse response $f_{i,j,h}$ can be written as
\begin{equation}
	s_{i,j,h} \times f_{i,j,h} \geq 0.
\end{equation}
For example, if $s_{i,j,h} = 1$, then this sign restriction implies that the $h$-step-ahead response of the $i$-th variable to the $j$-th structural shock is restricted to be non-negative. If $s_{i,j,h} = 0$, then the sign restriction is not imposed on this response. We define the sign restrictions set $\mathcal{S}$ to be the collection of $s_{i,j,h}$ for all $i,j,h$.

It is worth noting that in applications one typically imposes sign normalization restrictions on the $m$ structural shocks to facilitate their interpretation. These sign restrictions can be incorporated in our setup by including them in $\mathcal{S}$. For example, if one wishes to sign-normalize a monetary policy shock so that it is a \emph{contractionary} monetary policy shock, one could restrict the effect of the monetary policy shock (say, the $j$-th shock) on the policy rate (say, the $i$-th variable) to be non-negative on impact, i.e., $f_{i,j,0}\geq 0$. This can be done by setting $s_{i,j,0}  = +1.$

% A further extension of this identification scheme, called ranking restrictions, is considered in \citet{AD21}.

Now, we can formally define the admissible set with respect to the set of sign restrictions $\mathcal{S}$ and model parameters $(\bA,\vSigma)$:
\begin{align*}
	\mathcal{Q}(\bA,\vSigma,\mathcal{S}) = \{\bQ: \bQ  & \in \bO(n) \text{ and the impulse responses implied by } \bQ \text{ and } (\bA,\vSigma) \\ 
	& \text{ satisfy the restrictions in } \mathcal{S}\}.
\end{align*}

A popular algorithm to obtain draws uniformly from the admissible set $\mathcal{Q}(\bA,\vSigma,\mathcal{S})$ is given in \citet{RWZ10}. It is an accept-reject algorithm and is implemented as follows. First, obtain a draw $\bQ$ uniformly from the orthogonal group $\bO(n)$ (i.e., according to the Haar measure). Then, set $\bR = \bL\bQ$, where $\bL$ is the lower triangular Cholesky factor of $\vSigma$. If the impulse responses implied by $(\bA, \bR)$ satisfy all the restrictions in $\mathcal{S}$, then we accept $\bQ$ (it is easy to see that $\bQ\in \mathcal{Q}(\bA,\vSigma,\mathcal{S}))$; otherwise, we obtain another draw uniformly from $\bO(n)$ and repeat the procedure.

This algorithm is flexible and easy to implement and works well for a wide range of applications using small VARs. When the application requires a VAR that involves more than a dozen variables and restrictions, this algorithm tends to be computationally intensive, as it requires a large number of uniform draws from $\bO(n)$ to get each draw from the admissible set $\mathcal{Q}(\bA,\vSigma,\mathcal{S})$. When $n$ is large, this approach is simply computationally infeasible. 

% In addition, the computational cost of increases when $n$ increases. This is because obtaining a uniform draw from $\bO(n)$ requires $n^2$ standard normal random variables and a QR decomposition of an $n\times n$ matrix, which has computational complexity $\mathcal{O}(n^3)$. When $n$ is large, this approach is simply computationally infeasible. 

\subsection{A New Algorithm} \label{ss:new_algorithm}

For high-dimensional systems with a large number of sign restrictions, it is highly unlikely that any given uniform draw from $\bO(n)$, denoted as $\bQ\sim \distn{U}(\bO(n))$, would imply impulse responses that satisfy all the restrictions in $\mathcal{S}$. To make progress, we assume that $\mathcal{S} = \mathcal{S}_0$ where $\mathcal{S}_0$ collects sign restrictions that restrict only the signs of impulse responses at impact, i.e., $\mathcal{S}_0 = \{s_{i,j,0}: s_{i,j,0}\in\mathcal{S}\}.$ There are two key reasons to focus on the subset $\mathcal{S}_0$. First, there is often a strong consensus in economic theory about the signs of impulse responses at impact but not at longer horizons \citep[see, e.g.,][]{CP11}. Second, verifying the sign restrictions on impulse responses at impact is equivalent to verifying the signs of the elements in $\bR$, where $\bR=\bL\bQ$ and $\bL$ is the (lower triangular) Cholesky factor of $\vSigma$. As such, this verification can be done very quickly without computing impulse responses at all. Given a uniform draw $\bQ$, one can build a huge collection of equivalent draws (defined below) and search through this collection to obtain any members that satisfy all sign restrictions with trivial computations. 

More specifically, given $\bQ\sim \distn{U}(\bO(n))$, let $\mathcal{E}(\vSigma,\bQ)$ denote the set
\begin{align*}
	\mathcal{E}(\vSigma,\bQ) = \{ \bE&: \bE = \bL\bQ\bP\bD, \text{where } \bL \text{ is the Cholesky factor of } \vSigma, \bP \text{ is an } n\text{-dimensional } \\
	& \text{permutation matrix and } \bD \text{ is a diagonal matrix with elements } \pm 1 \}.
\end{align*}
In other words, $\mathcal{E}(\vSigma,\bQ)$ consists of all the permutations and sign switches of the columns of $\bL\bQ$. Since there are $n!$ permutation matrices of dimension $n$ and $2^n$ ways to construct an $n$ vector from the two values $\pm 1$, the cardinality of $\mathcal{E}(\vSigma,\bQ)$ is $2^n n!$.  

There are three key reasons to consider the set $\mathcal{E}(\vSigma,\bQ)$. First, since each column of $\bL\bQ$ can be viewed as the responses of the endogenous variables to a particular structural shock at impact, $\mathcal{E}(\vSigma,\bQ)$ includes all possible permutations and sign normalizations of the structural shocks represented by $\bQ$. That is, any member in $\mathcal{E}(\vSigma,\bQ)$ represents exactly the same structural shocks as $\bQ$---after relabeling the shocks and proper sign normalizations. Second, for any fixed $\bP$ or $\bD$ (respectively, a permutation matrix and a diagonal matrix with elements $\pm 1$), it is orthogonal. Therefore, the Haar measure is invariant under right multiplication of $\bP$ and $\bD$. Hence, $\bQ\bP\bD$ is a uniform draw from $\bO(n)$. Third, one can efficiently search through all the elements---all $2^n n!$ of them---in $\mathcal{E}(\vSigma,\bQ)$ to find those that satisfy all the restrictions in $\mathcal{S}_0$ (discussed below). Put differently, given each $\bQ\sim \distn{U}(\bO(n))$, we automatically obtain $2^n n!$ economically equivalent candidates with trivial additional computations. For $n=10$,  the number of orthogonal matrices that we sort through is about 3.7 billion. When $n=30$, the number is about $2.85\times 10^{41}$.

To distinguish two structural shocks, we require that they have signed impacts on at least two common endogenous variables. In addition, their impacts on one variable have the same sign, while their impacts on the other variable have opposite signs. More formally, we assume that $\mathcal{S}_0$ satisfies the following assumption:
\begin{assumption}\label{ass:S0} \rm For any $j\neq k,$ $j,k=1,\ldots, m$, there exist $i_1$ and $i_2$ such that $s_{i_1,j,0} = s_{i_1,k,0} \neq 0$ and $s_{i_2,j,0} = -s_{i_2,k,0} \neq 0$.
\end{assumption}

\begin{example}\label{ex:ass1}\rm Consider the following two sets of restrictions $\mathcal{S}_0^1$ and $\mathcal{S}_0^2$:
\[
	\mathcal{S}_0^1 = \begin{pmatrix} +1 & +1 & +1 \\ -1 & 0 & +1 \\ 0 & +1 & -1 \end{pmatrix}, \quad
	\mathcal{S}_0^2 = \begin{pmatrix} +1 & +1 & +1 \\ -1 & +1 & +1 \\ 0 & +1 & -1 \end{pmatrix}.
\]
$\mathcal{S}_0^1$ does not satisfy Assumption 1 because the first and second shocks (corresponding to the first and second columns) have signed impacts on only one common variable (the first variable). Intuitively, there is not enough information to separate these two structural shocks. For instance, a column of $\bR$ with $(1, -2, 3)'$ is consistent with both the first and second structural shocks. In contrast, $\mathcal{S}_0^2 $ satisfies Assumption 1: the first and second shocks have the same signed impact on the first variable but different signed impacts on the second variable; similarly for the first and third shocks; finally, the second and third shocks have the same signed impact on the first variable but different signed impacts on the third variable. Consequently, any column of $\bR$ can be consistent with at most one structural shock.
\end{example}

Next, we describe an efficient way to go through all the elements in $\mathcal{E}(\vSigma,\bQ)$ to locate those that satisfy all the restrictions in $\mathcal{S}_0$. Suppose that we have $n$ endogenous variables and we are interested in $m$ structural shocks. Let $\bT$ denote an $m\times n$ matrix such that $T_{ji}$, the $(j,i)$ element, is $+1$ if the $i$-th column of $\bR = \bL\bQ$ satisfies all the restrictions in $\mathcal{S}_0$ corresponding to $j$-th structural shock. If the negative of the $i$-th column of $\bR$ satisfies all the inequalities in $\mathcal{S}_0$ corresponding to $j$-th structural shock, set $T_{ji} = -1$; otherwise $T_{ji} = 0$. In other words, the $j$-th row of $\bT$ encodes all potential candidates among the columns of $\bR$ that can represent the $j$-th structural shocks (those have entries $\pm 1$). Therefore, if any row of $\bT$ contains all 0, then none of the elements in $\mathcal{E}(\vSigma,\bQ)$ satisfies all the restrictions in $\mathcal{S}_0$. In addition, by Assumption~\ref{ass:S0}, each column of $\bT$ has at most one $+1$ or $-1$---i.e., each column of $\bR$ can satisfy (or violate) all the restrictions of at most one structural shock.

\begin{example}\label{ex:T}\rm Suppose we have a set of restrictions to identify $m=2$ structural shocks using a VAR with $n=4$ variables:
\[
	\mathcal{S}_0^3 = \begin{pmatrix} +1 & +1 \\ +1 & -1 \\ 0 & 0 \\ 0 & 0 \end{pmatrix}.
\]
It is straightforward to verify that $\mathcal{S}_0^3$ satisfies Assumption 1. Further suppose we obtain two draws of $\bR$:
\[
	\bR^1 = \begin{pmatrix} 0.2 & 0.1 & -0.3 & 0.8 \\ 
	0.3 & 0.2 & -0.4 & 0.7 \\ 
	0.1 & -1.1& 1.2  & -0.4\\ 
	1.2 & 0.5& 0.5  & -1.2 \end{pmatrix}, \quad
	\bR^2 = \begin{pmatrix} 0.2 & -0.1 & -0.3 & 0.8 \\ 
	0.3 & 0.2 & -0.4 & 0.7 \\ 
	0.1 & -1.1& 1.2  & -0.4\\ 
	1.2 & 0.5& 0.5  & -1.2 \end{pmatrix}.
\]
The only difference between $\bR^1$ and $\bR^2$ is that their (1,2) elements have different signs.
Then, their corresponding $\bT^1$ and $\bT^2$ are:
\[
	\bT^1 = \begin{pmatrix} +1 & +1 & -1 & +1 \\ 0 & 0 & 0 & 0 \end{pmatrix}, \quad
	\bT^2 = \begin{pmatrix} +1 & 0 & -1 & +1 \\ 0 & -1 & 0 & 0 \end{pmatrix}.
\]
In other words, all columns of $\bR^1$ can potentially represent the first structural shock (e.g., after switching the signs of the elements in the third column), whereas none of the columns are consistent with the second structural shock. For $\bR^2$, the first, third and fourth columns are consistent with the first structural shock, while the second column is consistent with the second  structural shock. Finally, note that since $\mathcal{S}_0^3$ satisfies Assumption 1, each column of $\bT^1$ and $\bT^2$ has at most one $+1$ or $-1$.
\end{example}

It is important to note that to compute the matrix $\bT$, we only need to check each column of $\bR$ to see if all the relevant inequalities are all satisfied, all violated or neither, for each structural shock $j=1,\ldots,m$. Hence, it involves at most checking $mn^2$ inequalities to construct $\bT$, which can be done quickly. 

Let $\mathcal{E}(\vSigma,\bQ,\mathcal{S}_0)$ denote the subset of elements in $\mathcal{E}(\vSigma,\bQ)$ that satisfy all the restrictions in $\mathcal{S}_0$. In other words, $\mathcal{E}(\vSigma,\bQ,\mathcal{S}_0)$ consists of all the permutations and sign switches of the columns of 
$\bR = \bL\bQ$, where $\vSigma = \bL\bL'$, that satisfy all the restrictions in $\mathcal{S}_0$. Hence, any element in $\mathcal{E}(\vSigma,\bQ,\mathcal{S}_0)$ can be written as 
$\bR^* = \bR\bP\bD$ for some $n$-dimensional permutation matrix $\bP$ and diagonal matrix $\bD$ with elements $\pm 1$. And since $\bP$ and $\bD$ are orthogonal matrices, $\bR^* = \bR\bQ^*$ with 
$\bQ^* = \bP \bD \in\bO(n)$. 

\begin{algorithm}[H]
\caption{A new accept-reject algorithm to uniformly draw from the admissible set.}
\label{alg:ar}

\begin{enumerate}
	\item Sample a posterior draw of $(\bA, \vSigma),$ and obtain the lower triangular Cholesky factor $\bL$ of $\vSigma$ such that $\vSigma = \bL\bL'$.
	
	\item Sample $\bQ\sim \distn{U}(\bO(n))$. This can be done by sampling $\bZ=(Z_{ij})$, where 
	$Z_{ij}$ are iid $\distn{N}(0,1)$ random variables, and returning the orthogonal matrix $\bQ$
	from the QR decomposition of $\bZ$.
	
	\item Given $\bL$ and $\bQ$, construct $\bR = \bL\bQ$ and the associated $m\times n$ matrix $\bT$.
		
	\item If any row of $\bT$ contains all 0, then go back to Step 1; otherwise, let $\bR^*$ be an $n\times n$ zero matrix and complete the following steps: 
	\begin{enumerate}	
		\item For $j=1,\ldots, m$, construct the index set $S_j=\{i: T_{ji} = +1  \text { or } T_{ji} = -1, i=1,\ldots, n\} $ and sample an element uniformly from $S_j$, denoted as, $i_j$. 
		If $T_{ji_j} = +1$, set the $j$-th column of $\bR^*$ as the $i_j$-th column of $\bR$;
		if $T_{ji_j} = -1$, set the $j$-th column of $\bR^*$ as the negative of the $i_j$-th column of $\bR$.
		
		\item For $j=m+1,\ldots, n$, let $S_j = \{1,\ldots, n\}\backslash\{i_1,\ldots, i_{j-1}\}$ and sample an element uniformly from $S_j$, denoted as, $i_j$. With probability 1/2, set the $j$-th column of $\bR^*$ as the $i_j$-th column of $\bR$; otherwise set it as the negative of the $i_j$-th column of $\bR$.		
		
	\end{enumerate}	
	
	\item Return $\bA$ and $\bR^*$, which represents the structural shocks that satisfy all the restrictions in $\mathcal{S}_0$.
	
\end{enumerate}
\end{algorithm}

Given the matrix $\bT$, we can first determine whether or not $\mathcal{E}(\vSigma,\bQ,\mathcal{S}_0)$ is empty. If it is not, we then uniformly obtain an element from it as follows. Since each row of $\bT$ contains at least one $+1$ or $-1$, for each $j=1,\ldots, m$, we uniformly pick a column that has entries $+1$ or $-1$, say, $i_j$. And since each column contains at most one $+1$ or $-1$, we would not pick the same column twice. Given the sampled $i_1,\ldots, i_m$, we can reconstruct the element in $\mathcal{E}(\vSigma,\bQ,\mathcal{S}_0)$ that satisfies all the restrictions in $\mathcal{S}_0$. We summarize this algorithm in Algorithm~\ref{alg:ar}. 

%
%More specifically, given $\bQ\sim \distn{U}(\bO(n))$, let $\mathcal{E}(\vSigma,\bQ)$ denote the set
%\begin{align*}
	%\mathcal{E}(\vSigma,\bQ) = \{ \bE&: \bE = \bL\bQ\bP\bD, \text{where } \bL \text{ is the Cholesky factor of } \vSigma, \bP \text{ is an } n\text{-dimensional } \\
	%& \text{permutation matrix and } \bD \text{ is a diagonal matrix with elements } \pm 1 \}.
%\end{align*}
%In other words, $\mathcal{E}(\vSigma,\bQ)$ consists of all the permutations and sign switches of the columns of $\bL\bQ$. Since there are $n!$ permutation matrices of dimension $n$ and $2^n$ ways to construct an $n$ 

\begin{proposition}\label{prop:alg1} \rm  Under Assumption~\ref{ass:S0}, the output $\bR^*$ from Algorithm~\ref{alg:ar} satisfies $\bR^* = \bL\bQ^*$, where $\bQ^*$ is a uniform draw from $\bO(n)$ and $\bL\bL'=\vSigma$. In addition, $\bR^*$ satisfies all the restrictions in $\mathcal{S}_0$.
\end{proposition}

In other words, the proposed algorithm returns structural shocks that satisfy all the restrictions in $\mathcal{S}_0$ using a uniform draw $\bQ^*$ from the orthogonal group $\bO(n)$ such that $\bR^* = \bL\bQ^*$. The proof of the proposition is given in Appendix~A. 

It is worth noting that in Algorithm~\ref{alg:ar} we accept or reject each pair $(\bA,\vSigma)$ and $\bQ$ jointly, in contrast to the common practice of sampling $\bQ$ conditional on a particular draw of $(\bA,\vSigma)$ until all the restrictions in $\mathcal{S}_0$ are satisfied. As pointed out in \citet{ARSW24}, the latter approach has the unintended consequence of sampling from a different distribution other than the target.\footnote{Algorithm 1 does not check whether the admissible set $\mathcal{Q}(\mathbf{A},\vSigma, \mathcal{S})$ is empty or not, but this can be assessed using the procedure based on computing the Chebyshev center as proposed in \citet{AD21}. Since Algorithm 1 accepts or rejects $(\bA,\vSigma)$ and $\bQ$ jointly, if the admissible set $\mathcal{Q}(\mathbf{A},\vSigma, \mathcal{S})$ is empty, $\bT$ will have a zero row for any $\bQ$. Algorithm~\ref{alg:ar} then iterates and samples a different draw for $(\bA,\vSigma)$. In practice, it is much faster to compute $\bT$ than the Chebyshev center of $\mathcal{Q}(\mathbf{A},\vSigma, \mathcal{S})$.} It is also important to emphasize that Algorithm~\ref{alg:ar} returns one draw $\bR^*$, instead of multiple draws, from the set $\mathcal{E}(\vSigma,\bQ,\mathcal{S}_0)$ for each pair of $(\vSigma,\bQ)$. This is because any two draws from $\mathcal{E}(\vSigma,\bQ,\mathcal{S}_0)$ are dependent, as they differ by a permutation and sign switches of the columns. In contrast, given independent posterior draws of $(\bA,\vSigma)$, invoking Algorithm~\ref{alg:ar} multiple times to obtain multiple draws of $\bR^*$ ensures that these draws are independent.\footnote{
One can ensure that the Monte Carlo simulation error is a given fraction of the posterior standard deviation by choosing the number of admissible draws appropriately, and this choice does not depend on the dimension of the problem, provided that the retained admissible draws are independent. Specifically, given $R$ independent admissible draws, the Monte Carlo error of the posterior mean estimate is $1/\sqrt{R}$ as a fraction of the posterior standard deviation. In other words, the Monte Carlo errors of the posterior means obtained from 1,000 independent admissible draws are about 3.2\% of the corresponding posterior standard deviations.}

Assumption 1 is crucial for the implementation of Algorithm~1. Intuitively, its role is to ensure that one can distinguish any two structural shocks via only the responses of the variables at impact. This assumption allows us to use Algorithm~1 to uniformly obtain a draw from the set $\mathcal{E}(\vSigma,\bQ,\mathcal{S}_0)$ without enumerating all the elements in the set. Without this assumption, one would need to enumerate or label all the elements in $\mathcal{E}(\vSigma,\bQ,\mathcal{S}_0)$, which could be time-consuming if the set is large. While Assumption~1 is satisfied in many commonly-used identification schemes (which is vacuously satisfied when only one structural shock is identified), it would not be satisfied, e.g., when two shocks cause exactly the same responses at impact; an example is the identification of the TFP surprises vs TFP news shocks \citep{BP06}. In those cases, Algorithm 1 cannot be directly applied.

Nevertheless, one can still uniformly obtain a draw from $\mathcal{E}(\vSigma,\bQ,\mathcal{S}_0)$ by modifying Algorithm~1. More specifically, we can follow Steps 1-3 of the algorithm and construct the matrix $\bT$. Then, instead of Step~4,  we enumerate all the elements in $\mathcal{E}(\vSigma,\bQ,\mathcal{S}_0)$ and obtain a draw from the set uniformly. Below we provide a concrete example to illustrate the revised algorithm.

\begin{example}\label{ex:vio1}\rm Now consider a set of sign restrictions, $\mathcal{S}_0^4$, that does not satisfy Assumption~1, and a draw of $\bR$, denoted as $\bR^3$:
\[
	\mathcal{S}_0^4 = \begin{pmatrix} +1 & +1 \\ +1 & +1 \\ 0 & 0 \\ 0 & 0 \end{pmatrix}, 
	\quad 
	\bR^3 = \begin{pmatrix} 
	0.2 & -0.1 & 0.3 & 0.8 \\ 
	0.3 & 0.2 & -0.4 & 0.7 \\ 
	0.1 & -1.1& 1.2  & -0.4\\ 
	1.2 & 0.5& 0.5  & -1.2 \end{pmatrix}.
\]
Given $\mathcal{S}_0^4$ and $\bR^3$, the corresponding $\bT^3$ is:
\[
	\bT^3 = \begin{pmatrix} +1 & 0 & 0 & +1 \\ +1 & 0 & 0 & +1 \end{pmatrix}.	
\]
Note that it is no longer the case that each column of $\bT^3$ has at most one $+1$ or $-1$, and consequently, Algorithm 1 cannot be directly applied. Nevertheless, one can still enumerate all the elements in $\mathcal{E}(\vSigma,\bQ,\mathcal{S}_0)$ (i.e., all the permutations and sign switches of the columns of $\bR^3$ that satisfy all the restrictions in $\mathcal{S}_0^4$) using the matrix $\bT^3$. For example, one can use the first column of $\bR^3$ to represent the first structural shock and the fourth column to represent the second structural shock. Another possibility is to represent the first and second structural shocks using, respectively, the fourth and first columns of $\bR^3$. After enumerating all the elements in $\mathcal{E}(\vSigma,\bQ,\mathcal{S}_0)$, one can then uniformly obtain a draw from $\mathcal{E}(\vSigma,\bQ,\mathcal{S}_0)$.

\end{example}

\subsection{Extensions} \label{ss:extensions}

In this section we discuss how the proposed algorithm can be extended to handle some other commonly-used identification schemes. We start with the ranking restrictions of \citet{AD21}. In particular, consider the ranking restriction of the form $s_{i,j,k,l} f_{i,j,0} \geq s_{i,j,k,l} \lambda_{i,j,k,l} f_{k,l,0}$ for $s_{i,j,k,l}\in\{-1,0,1\}$ and $\lambda_{i,j,k,l}\geq 0$, where $f_{i,j,0}$ is the impulse response of the $i$-th variable from the $j$-th structural shock on impact, as defined in \eqref{eq:IR}. 

For example, if $i=k$, $s_{i,j,k,l}=1$ and $\lambda_{i,j,k,l}=1$, then this ranking restriction implies that the impact of the $j$-th structural shock on the $i$-th variable is at least as large as the impact of the $l$-th shock on the same variable. On the other hand, if $j=l$, $s_{i,j,k,l}=1$ and $\lambda_{i,j,k,l}=1$, then this ranking restriction implies that the response of the $i$-th variable to the $j$-th structural shock at least as large as the response of the $k$-th variable to the same shock. Furthermore, it is easy to see that the ranking restriction includes the sign restriction as a special case by setting $ \lambda_{i,j,k,l}=0$. 

Let $\mathcal{R}_0 =\{(s_{i,j,k,l},\lambda_{i,j,k,l}): i,k=1,\ldots, n, j,l=1,\ldots, m\}$ denote the ranking restrictions set on impact. We first consider the case where each ranking restriction involves only an individual structural shock (i.e., for $j\neq l$, $\lambda_{i,j,k,l}=0$); the general case will be discussed afterward. In addition, to ensure that the structural shocks are distinct, we impose some regularity conditions on $\mathcal{R}_0$. Intuitively, to distinguish two structural shocks, we require that either 1) they have signed impacts on at least two common endogenous variables, where on one variable they have the same sign and on the other they have opposite signs; or 2) the impact on a linear combination of two variables from one shock is positive whereas that from the other shock is negative. Formally, we assume $\mathcal{R}_0$ satisfies the following assumption:
\begin{assumption}\label{ass:R0} \rm For any $j\neq l,$ $j,l=1,\ldots, m$, at least one of the following conditions hold:
\begin{enumerate}

\item there exist $i_1$ and $i_2$ such that $s_{i_1,j,k_1,m_1} = s_{i_1,l,k_2,m_2} \neq 0$ and $s_{i_2,j,k_3,m_3} = -s_{i_2,l,k_4,m_4} \neq 0$ for some $k_1,k_2, k_3, k_4, m_1, m_2,m_3,m_4$, with $\lambda_{i_1,j,k_1,m_1} = \lambda_{i_1,l,k_2,m_2} = \lambda_{i_2,j,k_3,m_3} = \lambda_{i_2,l,k_4,m_4} = 0$ ;

% \item {\color{blue}there exists $i_1$ such that $s_{i_1,j,i_1,l} \neq 0$ and $\lambda_{i_1,j,i_1,l} > 0$;} -- Algorithm 1 cannot handle this condition because it would need to check the responses for two shocks
		
\item there exist $i_1$ and $i_2$ such that $s_{i_1,j,i_2,j} = -s_{i_1,l,i_2,l} \neq 0$ and $\lambda_{i_1,j,i_2,j} = \lambda_{i_1,l,i_2,l} > 0$.

\end{enumerate}
\end{assumption}

Condition 1 in Assumption~\ref{ass:R0} is essentially an extension of Assumption~\ref{ass:S0} to the case of ranking restrictions. For example, if $s_{i_1,j,k_1,j} = s_{i_1,l,k_2,l} = 1$ and $s_{i_2,j,k_3,j} = -s_{i_2,l,k_4,l} =1 $, then Condition 1 implies $f_{i_1,j,0}\geq 0 $, $f_{i_1,l,0}\geq 0 $, $f_{i_2,j,0}\geq 0 $, $f_{i_2,l,0}\leq 0 $. Condition 2 discriminates the two structural shocks by their different signed impacts on a linear combination of two variables.  For instance, if $s_{i_1,j,i_2,j} = -s_{i_1,l,i_2,l} = \lambda_{i_1,j,i_2,j} =  \lambda_{i_1,l,i_2,l} = 1$, then Condition 2 implies $f_{i_1,j,0} - f_{i_2,j,0} \geq 0 $ and $f_{i_1,l,0} - f_{i_2,l,0} \leq 0 $.

\begin{example}\label{ex:ass2}\rm Consider an example with $n=m=3$, and the set of ranking restrictions $\mathcal{R}_0$ is characterized by two 4-dimensional arrays or tensors $\mathbb{S} = (s_{i,j,k,l})$ and $\vLambda = (\lambda_{i,j,k,l})$, where the only non-zero elements in $\mathbb{S}$ and $\vLambda$ are:
\begin{align*}
	s_{1,1,1,1} &= 1, s_{1,2,1,1} = 1, s_{2,1,1,1} = 1, s_{2,2,1,1} = -1, \\
	s_{3,1,1,1} &= 1, s_{1,3,1,1} = 1, s_{3,3,1,1} = -1, \\	
	s_{1,1,1,2} &= 1, s_{2,1,2,3} = -1, s_{1,2,3,2} = 1, s_{1,3,3,3} = -1, 
	\lambda_{1,2,3,2} = 1, \lambda_{1,3,3,3} = 1.
\end{align*}
Then, $\mathcal{R}_0$ satisfies Assumption 2. More specifically, to distinguish the first and second structural shocks, note that $s_{1,1,1,1} = s_{1,2,1,1} = s_{2,1,1,1} = -s_{2,2,1,1}=1 $, so that Condition~1 is satisfied with  $i_1 = 1, i_2 = 2, j=1$ and $l=2$ (i.e., the responses of the first and second variables to the first shock are positive, whereas their responses to the second shock are positive and negative, respectively). Similarly, one can distinguish the first and third shocks since $s_{1,1,1,1} = s_{3,1,1,1} = s_{1,3,1,1} = -s_{3,3,1,1} = 1,$ and Condition 1 is satisfied with  $i_1 = 1, i_2 = 3, j=1$ and $l=3$ (i.e., the responses of the first and third variables to the first shock are positive, whereas their responses to the third shock are positive and negative, respectively). Finally, the remaining restrictions ensure that Condition 2 is satisfied with $i_1 = 1, i_2 = 3, j=2$ and $l=3$, and one can distinguish the second and third shocks. In particular, the restrictions imply that the second shock causes a larger response from the first variable than the third variable, whereas the third shock causes a smaller response from the first than the third variables. 
\end{example}

With Assumption~\ref{ass:R0}, we can easily adopt Algorithm~\ref{alg:ar} to obtain draws uniformly from the admissible set $\mathcal{E}(\vSigma,\bQ,\mathcal{R}_0)$. In fact, the only modification one needs is to replace the sign restrictions set $\mathcal{S}_0$ with the ranking restrictions set $\mathcal{R}_0$ in the construction of the matrix $\bT$. More specifically, we construct the $m\times n$ matrix $\bT$ as follows: set $T_{ji} = +1$ if the $i$-th column of $\bR = \bL\bQ$ satisfies all the restrictions in $\mathcal{R}_0$ corresponding to $j$-th structural shock. If the negative of the $i$-th column of $\bR$ satisfies all the inequalities in $\mathcal{S}_0$ corresponding to $j$-th structural shock, set $T_{ji} = -1$; otherwise $T_{ji} = 0$. As before, $\bT$ can be constructed with trivial computations. In addition, by Assumption~\ref{ass:R0}, each column of $\bT$ has at most one $+1$ or $-1$ since each column of $\bR$ can satisfy (or violate) all the restrictions of at most one structural shock. The rest of the steps in Algorithm~\ref{alg:ar} remain exactly the same. 

More generally, ranking restrictions include cases where two different structural shocks are involved. For example, one could impose $f_{i,j,0} \geq \lambda_{i,j,i,l} f_{i,l,0}$ for $ \lambda_{i,j,i,l}>1$, i.e., the response of the $i$-th variable to the $j$-th structural shock is larger than the response from the $l$-th structural shock. This type of restrictions can be accommodated by an extra accept-reject step. Specifically, one can first use Algorithm~\ref{alg:ar} to obtain $\bA$ and $\bR^*$, which represents the structural shocks that satisfy all the ranking restrictions in $\mathcal{R}_0$. If $\bR^*$ also satisfies the additional ranking restrictions, we accept it; otherwise, we obtain another draw of the pair $\bA$ and $\bR^*$ and repeat the procedure. Similarly, this approach can be applied to cases when one wishes to impose sign or ranking restrictions on longer-horizon impulse responses. That is, we use Algorithm~\ref{alg:ar} to obtain $\bA$ and $\bR^*$, and if the implied longer-horizon impulse responses satisfy the required sign and ranking restrictions, we accept it;  otherwise, we sample $\bA$ and $\bR^*$ again until a draw is accepted.

% Next, we discuss the case when $\mathcal{S}_1$ is non-empty. In that case, we can first construct $\mathcal{E}(\vSigma,\bQ,\mathcal{S}_0)$, the set of elements in $\mathcal{E}(\vSigma,\bQ)$ that satisfy all restrictions in $\mathcal{S}_0$. Then, we go through $\mathcal{E}(\vSigma,\bQ,\mathcal{S}_0)$ to find all those elements that also satisfy $\mathcal{S}_1$. 

% *** revise: lend themselves to similar speed-up (like ranking restrictions), but others would not (IRs at higher horizons). But we can always first find all the rotations that satisfy the sign restrictions, and within these draws we further search for those that satisfy additional restrictions. The new method substantially speeds up the first part but not the second. So the new method would only be useful if there are a large number of sign restrictions. 

We have so far focused on the so-called B-model \citep{Lutkepohl05} in which the reduced-form errors $\bu_t$ are linear combinations of the structural shocks $\bv_t$, i.e., $\bu_t = \bB_0\bv_t$, and the researcher has useful prior information on $\bB_0$. The proposed approach can be adapted to handle the case when the researcher has useful prior information on $\bA_0=\bB_0^{-1}$ instead. More specifically, suppose we wish to impose sign restrictions on the elements of $\bA_0$. Note that 
$\bA_0'\bA_0 = \vSigma^{-1}$. Let $\overline{\bL}$ denote the \textit{upper} Cholesky factor of $\vSigma^{-1}$ so that $\vSigma^{-1} = \overline{\bL}'\overline{\bL}.$ It is clear that if we let $\overline{\bR} = \bQ \overline{\bL}$, where  $\bQ$ is an orthogonal matrix, we have $\overline{\bR}'\overline{\bR} =  \vSigma^{-1}$. Therefore, Algorithm~1 can be directly applied to this case by simply replacing $\bL$ and $\bR$ with, respectively, $\overline{\bL}$ and $\overline{\bR}$.

\section{Comparison of Computational Efficiency} \label{s:comparison}

The goal of this section is to demonstrate the empirical performance of the proposed algorithm in a variety of settings compared to other benchmarks such as the algorithms of \citet{RWZ10} and \citet{Read22}. In the first subsection, we consider an empirical example that involves a 15-variable VAR with over 40 sign and ranking restrictions. We compare both the speed and the estimated impulse responses against a benchmark. In the second subsection, we demonstrate how the proposed algorithm can be applied in settings with dynamic sign restrictions by replicating the application in \citet{Uhlig05}. In the last subsection, we further compare the computational efficiency of the proposed approach relative to \citet{RWZ10} along a few model dimensions.

\subsection{An Illustration of a 15-Variable VAR} \label{ss:illustration}

We first illustrate the empirical performance of the proposed algorithm using a 15-variable VAR with over 40 sign and ranking restrictions to identify 5 structural shocks. As a comparison, we also use the algorithm of \citet{RWZ10} to uniformly sample orthogonal matrices from the admissible set and compute the impulse responses.

More specifically, \citet{FRS19} first use a 6-variable VAR to identify 5 structural shocks---demand, supply, monetary, investment and financial shocks---using a set of sign and ranking restrictions on the contemporaneous impact matrix. \citet{chan22} augments their 6-variable system with 9 additional variables and sign restrictions. The variables and the structural restrictions are given in Table~\ref{tab:sign}. All rows except the fourth present the sign restrictions on the contemporaneous impact matrix. The fourth row represents ranking restrictions: the entries denote the signs of the differential impacts on investment and output from each structural shock. For example, $-1$ in the demand column indicates that the impact from demand shocks on investment is smaller than the impact on output. 

It is straightforward to see that this set of sign and ranking restrictions satisfies Assumption~\ref{ass:R0}. More specifically, supply and monetary shocks can be distinguished from other shocks using Condition 1 in Assumption~\ref{ass:R0}. In addition, demand shocks have a negative impact on the difference between investment and output, whereas the impacts from investment and financial shocks are positive. Hence, demand shocks can be distinguished from the other two shocks using Condition 2.

\begin{table}[H]
\centering
\caption{Sign restrictions, ranking restrictions and identified shocks for the 15-variable VAR.} \label{tab:sign}
\begin{tabular}{lccccc}
\hline\hline
	&	Supply	&	Demand	&	Monetary	&	Investment	&	Financial	\\ \hline
	GDP	&	$+1$	&	$+1$	&	$+1$	&	$+1$	&	$+1$	\\
	\rowcolor{lightgray}
	GDP deflator	&	$-1$	&	$+1$	&	$+1$	&	$+1$	&	$+1$	\\
	3-month tbill rate	&	0	&	$+1$	&	$-1$	&	$+1$	&	$+1$	\\
	\rowcolor{lightgray}
	Investment/GDP	&	0	&	$-1$	&	0	&	$+1$	&	$+1$	\\
	S\&P 500	&	$+1$	&	0	&	0	&	$-1$	&	$+1$	\\
	\rowcolor{lightgray}
	Spread	&	0	&	0	&	0	&	0	&	0	\\
	Spread 2	&	0	&	0	&	0	&	0	&	0	\\
	\rowcolor{lightgray}
	Credit/Real estate value	&	0	&	0	&	0	&	0	&	0	\\
	Mortgage rates	&	0	&	0	&	0	&	0	&	0	\\
	\rowcolor{lightgray}
	CPI	&	$-1$	&	$+1$	&	$+1$	&	$+1$	&	$+1$	\\
	PCE	&	$-1$	&	$+1$	&	$+1$	&	$+1$	&	$+1$	\\
	\rowcolor{lightgray}
	employment	&	0	&	0	&	0	&	0	&	0	\\
	Industrial production	&	$+1$	&	$+1$	&	$+1$	&	$+1$	&	$+1$	\\
	\rowcolor{lightgray}
	1-year tbill rate	&	0	&	$+1$	&	$-1$	&	$+1$	&	$+1$	\\
	DJIA	&	$+1$	&	0	&	0	&	$-1$	&	$+1$	\\ \hline\hline
\end{tabular}
{\raggedright \footnotesize{Note: All restrictions are imposed on the response of a particular variable, except for investment/GDP, in which restrictions are imposed on linear inequalities of two responses.} \par}
\end{table}

As a benchmark, we use the algorithm of \citet{RWZ10} to uniformly sample orthogonal matrices in conjunction with the posterior sampler of \citet{chan22} designed for large VARs to obtain posterior draws of the model parameters. This approach requires approximately 3.6 billion draws from $\distn{U}(\bO(n))$ to obtain 1,000 admissible draws, and the estimation takes about 6 days on a standard desktop. In contrast, the new algorithm requires only about 31,000 draws from $\distn{U}(\bO(n))$ to obtain 1,000 admissible draws, and the entire exercise takes about 16 seconds.\footnote{In this illustration there are far more variables than identified structural shocks. One might wonder how the proposed algorithm behaves in applications when the difference between the numbers of variables and structural shocks is small, as there are fewer `free' columns to potentially represent structural shocks. It is expected in those cases that the proposal algorithm will take more candidate draws to obtain one admissible draw. However, the speed-up relative to the benchmark remains the same---the cardinality of $\mathcal{E}(\vSigma, \bQ)$ is still $2^n n\!$. As an example, we replicate the original application in \citet{FRS19} that uses a 6-variable VAR to identify 5 structural shocks and one `residual' shock. The benchmark requires over 9 million candidate posterior draws and takes 4 minutes to obtain 1,000 admissible draws, whereas the proposed algorithm requires only 2 seconds and takes 21,000 candidate posterior draws.}

Next, we empirically verify that the impulse responses obtained from the two algorithms are the same. In particular, Figure \ref{fig:IRF_finc_15} reports the impulse responses of 6 variables to a one-standard-deviation financial shock, obtained using the algorithm of \citet{RWZ10}, and Figure~\ref{fig:IRF_finc_15_newalg} reports those from the proposed algorithm. As expected, the impulse responses obtained using the two algorithms are identical. Thus, these results highlight the utility of the proposed algorithm: it provides the same impulse responses but is several orders of magnitude more efficient than the benchmark.

\begin{figure}[H]
    \centering
   \includegraphics[width=.75\textwidth]{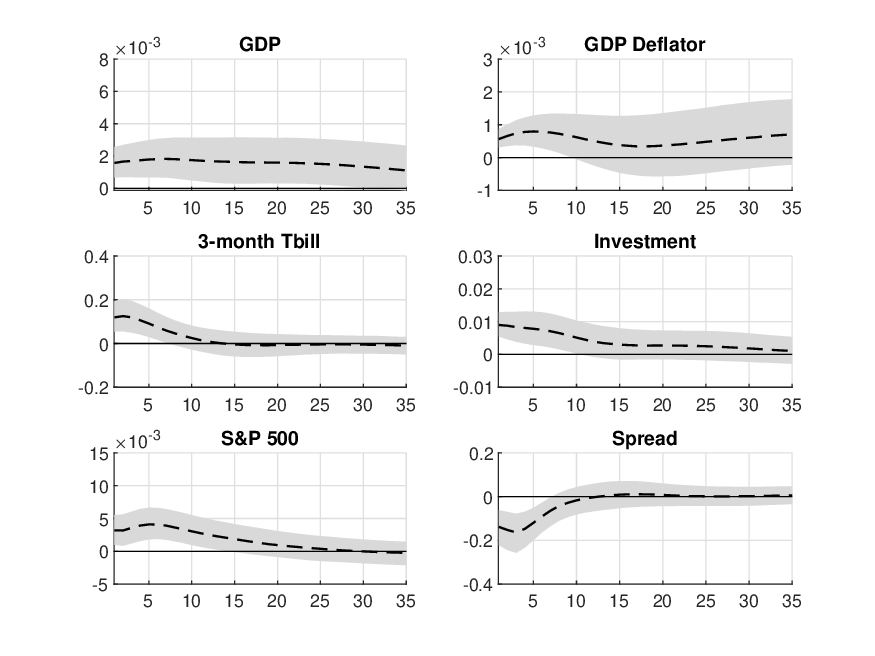}
   \caption{Impulse responses from a 15-variable VAR to a one-standard-deviation financial shock, obtained using the algorithm of \citet{RWZ10}.}
   \label{fig:IRF_finc_15}	
\end{figure}

\begin{figure}[H]
    \centering
   \includegraphics[width=.75\textwidth]{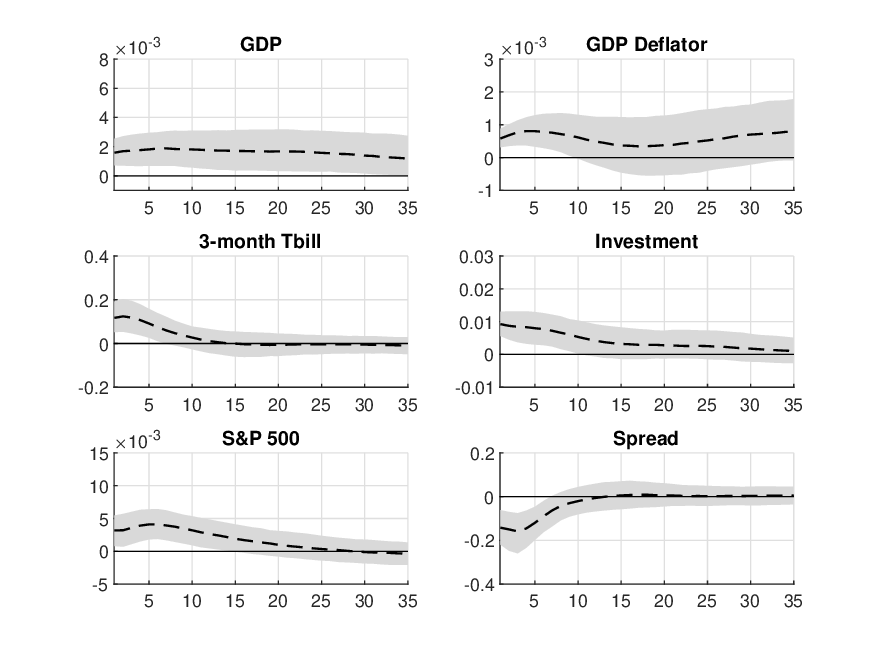}
   \caption{Impulse responses from a 15-variable VAR to a one-standard-deviation financial shock, obtained using the proposed algorithm described in Algorithm~\ref{alg:ar}.}
   \label{fig:IRF_finc_15_newalg}	
\end{figure}

\subsection{Dynamic Sign Restrictions: Replication of \citet{Uhlig05}}

Next, we demonstrate how the proposed algorithm can be applied in settings with both static and dynamic sign restrictions. More specifically, we consider the application in \citet{Uhlig05} that uses a 6-variable VAR to identify the monetary policy shock. The 6 monthly variables are industrial production, CPI, a commodity index (S\&P GSCI), total reserves, nonborrowed reserves and the effective federal funds rate. All variables are sourced from the FRED database maintained by the Federal Reserve Bank of St. Louis, except for the commodity index, which is obtained from S\&P Capital IQ. The federal funds rate is not transformed, whereas all other variables are in log. The sample period is from December 1969 to December 2007. 

To identify the monetary policy shock, we follow \citet{Uhlig05} and assume that the responses of prices and nonborrowed reserves are nonpositive and the responses of federal funds rate are nonnegative on impact and for the first 5 months. To implement the proposed approach to this setting with dynamic restrictions, we first use Algorithm 1 to obtain admissible draws that satisfy the required sign restrictions on impact. We then check if the impulse responses also satisfy the required sign restrictions for the first 5 months. If so, we accept the draw; otherwise, we obtain another admissible draw using Algorithm~1 and repeat the procedure until a draw is accepted. As a comparison, we also implement the algorithms of \citet{RWZ10} and \citet{Read22}.\footnote{We adapt the code provided by \citet{Read22} to our setting. In particular, we run his Algorithm 2 to draw uniformly the first column of $\bQ$ given one set of posterior draws, after using Algorithm 1 to check whether the admissible set is empty.} Table~\ref{tab:comparison} reports the computation time to obtain 5,000 posterior draws using the three algorithms.

% *** something to add: also compare to the algorithm in \citet{Read22}. Maybe add a little table to document the computation time of the 3 algorithms. Add a two by three panel to illustrate that all the impulse responses are the same across the algorithms.

\begin{table}[H]
	\centering 
	\caption{The computation time (in minutes) to obtain 5,000 posterior draws using the proposed method and the algorithms of \citet{RWZ10} (RWZ) and \citet{Read22} (Read).}
	\label{tab:comparison}
\begin{tabular}{lccc}
\hline \hline
       & Proposed method & RWZ & Read  \\ \hline
Time   & 1.05 & 14.27 & 3.07            \\ \hline\hline
\end{tabular}
\end{table}

To get 5,000 admissible draws that satisfy all the static and dynamic sign restrictions, the approach of \citet{RWZ10} takes about 14 minutes, whereas the algorithm of \citet{Read22} takes about three minutes. In contrast, the proposed algorithm is substantially faster than both approaches and takes only one minute. To verify that all three algorithms generate from the same target distribution, we compute the dynamic responses of the 6 variables to the identified monetary policy shock. The impulse-response functions are provided in Appendix B. The three methods give virtually identical impulse responses; they are also similar to those reported in \citet{Uhlig05}.

\subsection{Computational Efficiency in High-Dimensional Settings} \label{ss:MC}

The purpose of this subsection is to compare the computational efficiency of the proposed algorithm relative to \citet{RWZ10} in high-dimensional setting using simulated data. More specifically, we generate datasets with different numbers of variables ($n=10,30,50)$ and structural shocks ($m=5,8$), while fixing the sample size $T=200$ and lag length $p=5$ for all simulations.

For each $(n,m)$ combination, we generate a dataset from the VAR in \eqref{eq:VAR}-\eqref{eq:shock} as follows. First, we draw the intercepts independently from the uniform distribution on the interval $(-1,1)$, i.e., $\distn{U}(-1, 1)$. For the VAR coefficients, the diagonal elements of the first VAR coefficient matrix are iid $\distn{U}(0,0.5)$ and the off-diagonal elements are from $\distn{U}(-0.2,0.2)$; all other elements of the $j$-th ($j > 1$) VAR coefficient matrices are iid $\distn{N}(0,0.1^2/j^2).$ Finally, to construct the impact matrix $\bB_0$, we first draw the diagonal elements from iid $\distn{U}(0.5,1.5)$, and the off-diagonal elements from iid $\distn{N}(0,1)$. We then store them and change the signs of the elements in $\bB_0$ to match the set of restrictions specified in each case. 

Given a dataset, we then estimate the model using the proposed algorithm and the benchmark, together with the direct posterior sampler of \citet{chan22} designed for large VARs. Each algorithm is run for 10,000 seconds, and we record the numbers of posterior draws and admissible draws (i.e., those posterior draws that satisfy all the structural restrictions). The results are reported in Table~\ref{tab:postdraws}. The top panel refers to the case where only sign restrictions are used (and the set of sign restrictions satisfies Assumption 1); the middle panel considers the case where three additional ranking restrictions are added; lastly, the bottom panel considers the case with both static sign restrictions 
(as in the top panel) and dynamic sign restrictions (on the first shock at horizon $H=1$).

The results show that as the number of sign restrictions increases, the number of admissible draws obtained---given a fixed number of candidate draws---decreases for both algorithms. This is not surprising because the computational bottleneck of obtaining admissible draws lies in the fact that the set $\mathcal{Q}(\mathbf{A},\vSigma, \mathcal{S})$ becomes thinner when more restrictions are imposed. In many cases with large $n$ or large $m$, the sampling efficiency of the benchmark deteriorates so quickly that it becomes infeasible. In contrast, the proposed method remains capable of obtaining a large number of admissible draws for large $n$ and $m$ in a reasonable amount of time. Finally, the bottom panel shows that both algorithms slow down considerably when dynamic restrictions are added. Nevertheless, the proposed algorithm can still generate a reasonable number of admissible draws in such hard cases.

\begin{table}[H]
\centering 
\caption{Numbers of posterior draws (in millions) and admissible draws obtained for an $n$-variable VAR with $m$ shocks within 10,000 seconds using the proposed method and the algorithm of \citet{RWZ10} (RWZ).} 
\label{tab:postdraws}
\begin{tabular}{lllrrr}
\hline \hline
\multicolumn{6}{l}{Top panel: sign restrictions only}  \\  
        & &  & $n=10$               & $n=30$               & $n=50$              \\ 
        $m=5$          &                   & $\#$ restrictions & 25              & 35              & 40      \\ \cline{3-6} 
      & RWZ         & Posterior draws ($\times 10^6$)  & 240   & 35  & 14 \\
      &             & Admissible draws   & 1,033   & 0    & 0  \\
      & Proposed method  & Posterior draws ($\times 10^6$)  & 12 & 12 & 9      \\
      &             & Admissible draws   & 489,030 & 166,590 & 14,099 \\ \cline{2-6} 
$m=8$          &                   & $\#$ restrictions  & 40              & 50              & 60             \\ \cline{3-6} 
           & RWZ               & Posterior draws ($\times 10^6$)  & 232  & 36  & 14\\
           &                   & Admissible draws & 0  & 0  & 0    \\
           &  Proposed method        & Posterior draws ($\times 10^6$)  & 34 & 14 & 7     \\
           &                   & Admissible draws & 266,280 & 57,804 & 2,460 \\ \hline
\multicolumn{6}{l}{Middle panel: 3 additional ranking restrictions} \\ 
&  & & $n=10$               & $n=30$               & $n=50$          \\
$m=5$ &  & $\#$ restrictions    & 28              & 38              & 43             \\ \cline{3-6} 
           & RWZ               & Posterior draws ($\times 10^6$)  & 258  & 34  & 13\\
           &                   & Admissible draws & 115  & 0  & 0 \\
           &  Proposed method  & Posterior draws ($\times 10^6$)  & 47  & 18 & 8\\
           &                   & Admissible draws & 310,970 & 14,230 & 1,214  \\ \cline{2-6} 
$m=8$      &                   & $\#$ restrictions    & 43              & 53              & 63             \\ \cline{3-6} 
           & RWZ               & Posterior draws ($\times 10^6$)   & 260 & 34 & 14 \\
           &                   & Admissible draws & 0 & 0 & 0 \\
           & Proposed method   & Posterior draws ($\times 10^6$)   & 37  & 13 & 6 \\
           &                   & Admissible draws & 99,607 & 2,525 & 1,000 \\ \hline
\multicolumn{6}{l}{Bottom panel: additional dynamic restrictions for the first shock} \\ 
&  & & $n=10$               & $n=30$               & $n=50$          \\
$m=5$ &  & $\#$ restrictions    & 29              & 31              & 45             \\ \cline{3-6} 
           & RWZ               & Posterior draws ($\times 10^6$)  & 81  & 13  & 3\\
           &                   & Admissible draws & 375  & 0  & 0 \\
           &  Proposed method  & Posterior draws ($\times 10^6$)  & 15  & 4 & 2\\
           &                   & Admissible draws & 229,920 & 733 & 316  \\ \cline{2-6} 
$m=8$      &                   & $\#$ restrictions    & 45              & 55              & 66             \\ \cline{3-6} 
           & RWZ               & Posterior draws ($\times 10^6$)   & 67 & 14 & 3 \\
           &                   & Admissible draws & 0 & 0 & 0 \\
           & Proposed method   & Posterior draws ($\times 10^6$)   & 5  & 3 & 2 \\
           &                   & Admissible draws & 2499 & 416 & 247 \\
\hline\hline
\end{tabular}
\end{table}

%However, it is also clear that the proposed method is much more efficient compared to the benchmark. Furthermore, when $n$ or $m$ grows, the differences in performance between the two methods become more apparent. In fact, 

\section{A 35-Variable VAR of the US Economy} \label{s:application}

To showcase the usefulness of the proposed algorithm, we consider an application that involves a 35-variable VAR with sign and ranking restrictions to identify 8 structural shocks, namely, demand, investment, financial, monetary policy, government spending, technology, labor supply and wage bargaining. The list includes many standard macroeconomic and financial variables, such as national accounts variables, various inflation indexes and interest rates, labor market variables, oil and stock prices. These variables are broadly similar to those used in \citet{Crumpetal21} and are closely monitored by the Federal Reserve Staff and professional forecasters.

There are several reasons in favor of using a large set of macroeconomic and financial variables in structural analysis. First, a large system provides a convenient and unified framework to investigate the impacts of multiple structural shocks simultaneously. In particular, it allows the researcher to tease out the impacts of different structural shocks---such as different types of demand and supply shocks---and their individual contributions to macroeconomic fluctuations. 

Second, it mitigates the concern of informational deficiency of using a limited information set, as pointed out in a series of influential papers by \citet{HS91} and \citet{LR93, LR94}. By using a larger set of relevant variables, one can close the gap between the set of variables considered by the economic agent and that considered by the econometrician, thus alleviating the concern of non-fundamentalness \citep[see, e.g.,][for a recent review]{gambetti21}. 

Third, as argued in \citet{LMW22}, the mapping from variables in an economic model to the data is typically not unique. For example, one could match the economic variable `inflation' to data based on the CPI, PCE, or the GDP deflator. One natural way to avoid an arbitrary choice is to include multiple data series corresponding to the same economic variable in the analysis.

The list of variables and the structural restrictions are given in Table~\ref{tab:ranking}.\footnote{Most variables are transformed by taking logs and multiplying 100, while others such as interest rates and unemployment rates are not transformed and are in percentages.} The top part of the table lists the sign restrictions whereas the lower part lists the ranking restrictions. For example, the row labeled `Government spending/GDP' lists the signs of the differences in impacts on government spending and GDP from each structural shock. In particular, the $+1$ in the government spending column indicates that the impact from government spending shocks on government spending is larger than the impact on GDP. It can be easily verified that the set of restrictions in Table~\ref{tab:ranking} satisfies Assumption~\ref{ass:R0}.

\begin{table}[H]
\centering
\caption{Sign restrictions, ranking restrictions and identified shocks for the 35-variable VAR.} \label{tab:ranking}
\resizebox{\textwidth}{!}{\begin{tabular}{lcccccccc}
\hline\hline
Sign restrictions	&	Demand 	&	Investment	&	Financial	&	Monetary	&	Government	&	Technology	&	Labor &	Wage \\ 
	&		&		&		&		&	spending	&	&	supply	&	bargaining	\\ \hline
GDP	&	$+1$	&	$+1$	&	$+1$	&	$-1$	&	$+1$	&	$+1$	&	$+1$	&	$+1$	\\
\rowcolor{lightgray}
Personal consumption expenditure	&	0	&	0	&	0	&	0	&	0	&	$+1$	&	0	&	0	\\
Residential investment	&	0	&	0	&	0	&	0	&	0	&	0	&	0	&	0	\\
\rowcolor{lightgray}
Nonresidential investment	&	0	&	0	&	0	&	0	&	0	&	$+1$	&	0	&	0	\\
Exports	&	0	&	0	&	0	&	0	&	0	&	0	&	0	&	0	\\
\rowcolor{lightgray}
Imports	&	0	&	0	&	0	&	0	&	0	&	0	&	0	&	0	\\
Government spending	&	0	&	0	&	0	&	0	&	$+1$	&	0	&	0	&	0	\\
\rowcolor{lightgray}
Federal budget surplus/deficit	&	0	&	0	&	0	&	0	&	$-1$	&	0	&	0	&	0	\\
Federal tax receipts	&	0	&	0	&	0	&	0	&	$+1$	&	0	&	0	&	0	\\
\rowcolor{lightgray}
GDP deflator	&	$+1$	&	$+1$	&	$+1$	&	$-1$	&	$+1$	&	$-1$	&	$-1$	&	$-1$	\\
PCE index	&	$+1$	&	$+1$	&	$+1$	&	$-1$	&	$+1$	&	$-1$	&	$-1$	&	$-1$	\\
\rowcolor{lightgray}
PCE index less food \& energy	&	$+1$	&	$+1$	&	$+1$	&	$-1$	&	$+1$	&	$-1$	&	$-1$	&	$-1$	\\
CPI index	&	$+1$	&	$+1$	&	$+1$	&	$-1$	&	$+1$	&	$-1$	&	$-1$	&	$-1$	\\
\rowcolor{lightgray}
CPI index less food \& energy	&	$+1$	&	$+1$	&	$+1$	&	$-1$	&	$+1$	&	$-1$	&	$-1$	&	$-1$	\\
Hourly wage	&	0	&	0	&	0	&	0	&	0	&	$+1$	&	$-1$	&	$-1$	\\
\rowcolor{lightgray}
Labor productivity	&	0	&	0	&	0	&	0	&	0	&	$+1$	&	0	&	0	\\
Utilization-adjusted TFP	&	0	&	0	&	0	&	0	&	0	&	$+1$	&	0	&	0	\\
\rowcolor{lightgray}
Employment	&	0	&	0	&	0	&	$-1$	&	0	&	0	&	0	&	0	\\
Unemployment rate	&	$-1$	&	$-1$	&	$-1$	&	$+1$	&	$-1$	&	$-1$	&	$+1$	&	$-1$	\\
\rowcolor{lightgray}
Industrial production index	&	$+1$	&	$+1$	&	$+1$	&	$-1$	&	0	&	0	&	0	&	0	\\
Capacity utilization	&	$+1$	&	$+1$	&	$+1$	&	$-1$	&	0	&	0	&	0	&	0	\\
\rowcolor{lightgray}
Housing starts	&	0	&	0	&	0	&	0	&	0	&	0	&	0	&	0	\\
Disposable income	&	0	&	0	&	0	&	0	&	0	&	0	&	0	&	0	\\
\rowcolor{lightgray}
Consumer sentiment	&	0	&	0	&	0	&	0	&	0	&	0	&	0	&	0	\\
Fed funds rate	&	$+1$	&	$+1$	&	$+1$	&	$+1$	&	$+1$	&	0	&	0	&	0	\\
\rowcolor{lightgray}
3-month tbill rate	&	$+1$	&	$+1$	&	$+1$	&	$+1$	&	$+1$	&	0	&	0	&	0	\\
2-year tnote rate	&	0	&	0	&	0	&	$+1$	&	0	&	0	&	0	&	0	\\
\rowcolor{lightgray}
5-year tnote rate	&	0	&	0	&	0	&	$+1$	&	0	&	0	&	0	&	0	\\
10-year tnote rate	&	0	&	0	&	0	&	$+1$	&	0	&	0	&	0	&	0	\\
\rowcolor{lightgray}
Prime rate	&	$+1$	&	$+1$	&	$+1$	&	$+1$	&	$+1$	&	0	&	0	&	0	\\
Aaa corporate bond yield	&	0	&	0	&	0	&	$+1$	&	0	&	0	&	0	&	0	\\
\rowcolor{lightgray}
Baa corporate bond yield	&	0	&	0	&	0	&	$+1$	&	0	&	0	&	0	&	0	\\
Trade-weighted US\$ index	&	0	&	0	&	0	&	0	&	0	&	0	&	0	&	0	\\
\rowcolor{lightgray}
S\&P 500	&	0	&	$-1$	&	$+1$	&	$-1$	&	0	&	0	&	0	&	0	\\
Spot oil price	&	0	&	0	&	0	&	0	&	0	&	0	&	0	&	0	\\ \hline
Ranking restrictions	&		&		&		&		&		&		&		&		\\ \hline
Nonresidential investment/GDP	&	$-1$	&	$+1$	&	$+1$	&	0	&	0	&	0	&	0	&	0	\\
\rowcolor{lightgray}
Government spending/GDP	&	$-1$	&	$-1$	&	$-1$	&	0	&	$+1$	&	0	&	0	&	0	\\
 \hline\hline
\end{tabular}}
\end{table}

For large-$n$ systems it becomes necessary to regularize the large number of VAR coefficients (e.g., a 35-variable VAR with 5 lags has 6,125 VAR coefficients). Here we use the asymmetric conjugate prior and the direct sampling approach proposed in \citet{chan22} to obtain posterior draws from the 35-variable VAR. Unlike the conventional natural conjugate prior, which does not permit cross-variable shrinkage (i.e., asymmetric shrinkage of own lags and lags of other variables), this new prior is more flexible and can be used, e.g., to shrink coefficients on the other variables’ lags more aggressively to zero. At the same time, this new prior maintains many useful analytical results of the traditional natural conjugate prior, such as a closed-form expression of the marginal likelihood, and allows direct, independent sampling from the posterior distribution instead of using MCMC methods. The key step in formulating this new prior is to reparameterize the VAR in a recursive structural form with a diagonal error covariance matrix. Then, one can show that if the reduced-form error covariance matrix has a standard inverse-Wishart prior, then the implied prior on the structural-form impact matrix and error variances is a product of normal-inverse-gamma densities, which is conjugate for the likelihood.

Using this asymmetric conjugate prior, we obtain the values of the optimal shrinkage hyperparameters on the VAR coefficients by maximizing the marginal likelihood of the model. Then, we use Algorithm~\ref{alg:ar} to obtain 1,000 admissible draws that satisfy all the sign and ranking restrictions. For this 35-variable VAR with over 100 sign and ranking restrictions, the entire exercise takes about 14 minutes and requires 557,000 draws from $\distn{U}(\bO(n))$. 

Figures \ref{fig:IRF_35_demand}--\ref{fig:IRF_35_financial}	report the impulse responses of 6 selected variables to the (one-standard-deviation) demand, investment and financial shocks. As expected, these demand-type structural shocks raise output, short-term interest rate and inflation, while lowering both unemployment rate and real wage, at least in the short-run. Compared to the generic demand shock, both investment and financial shocks have a more substantive impact on nonresidential investment.

\begin{figure}[H]
    \centering
   \includegraphics[width=.8\textwidth]{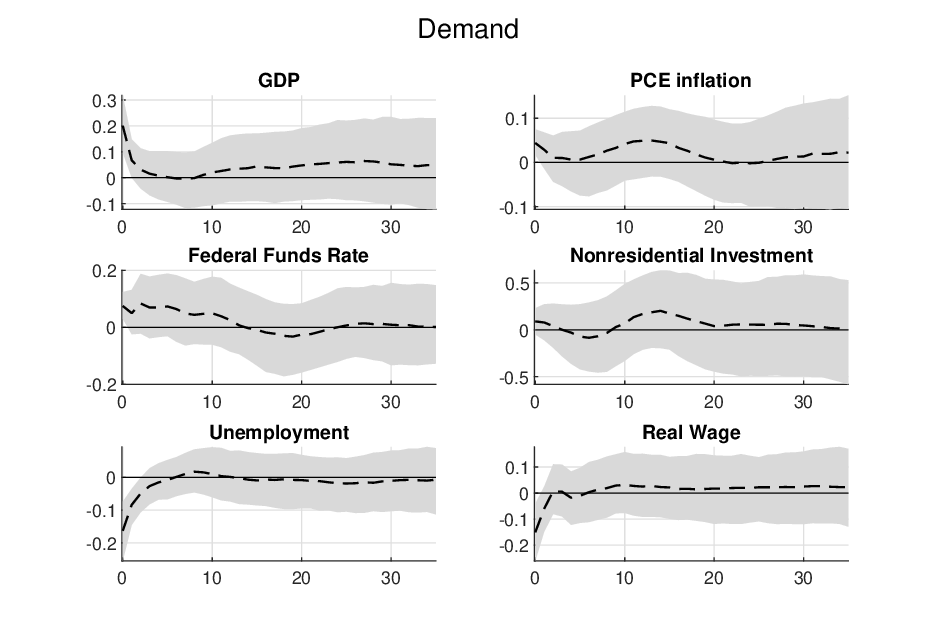}
   \caption{Impulse responses from a 35-variable VAR to a one-standard-deviation demand shock.}
   \label{fig:IRF_35_demand}	
\end{figure}

\begin{figure}[H]
    \centering
   \includegraphics[width=.8\textwidth]{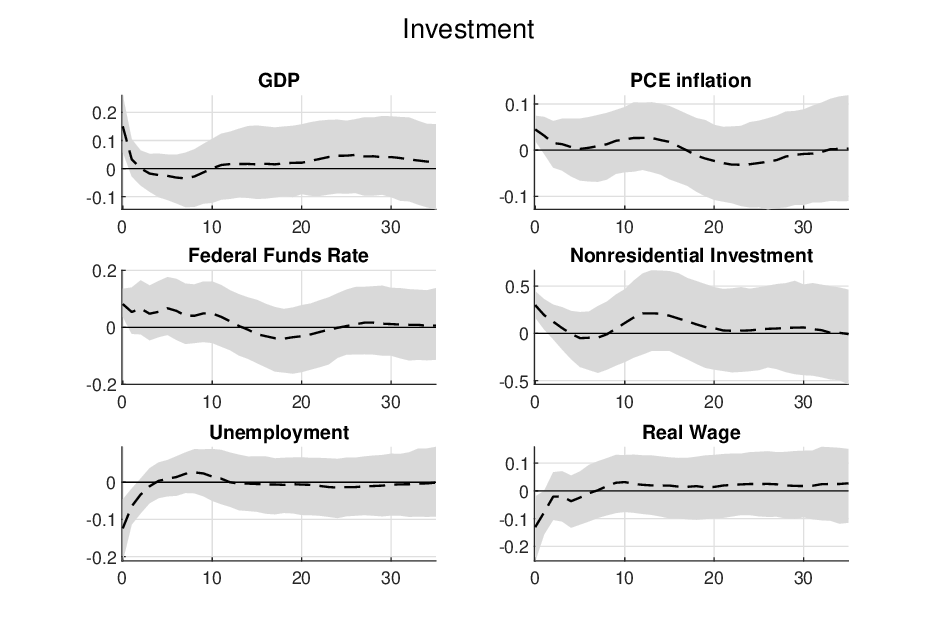}
   \caption{Impulse responses from a 35-variable VAR to a one-standard-deviation investment shock.}
   \label{fig:IRF_35_investment}	
\end{figure}

\begin{figure}[H]
    \centering
   \includegraphics[width=.8\textwidth]{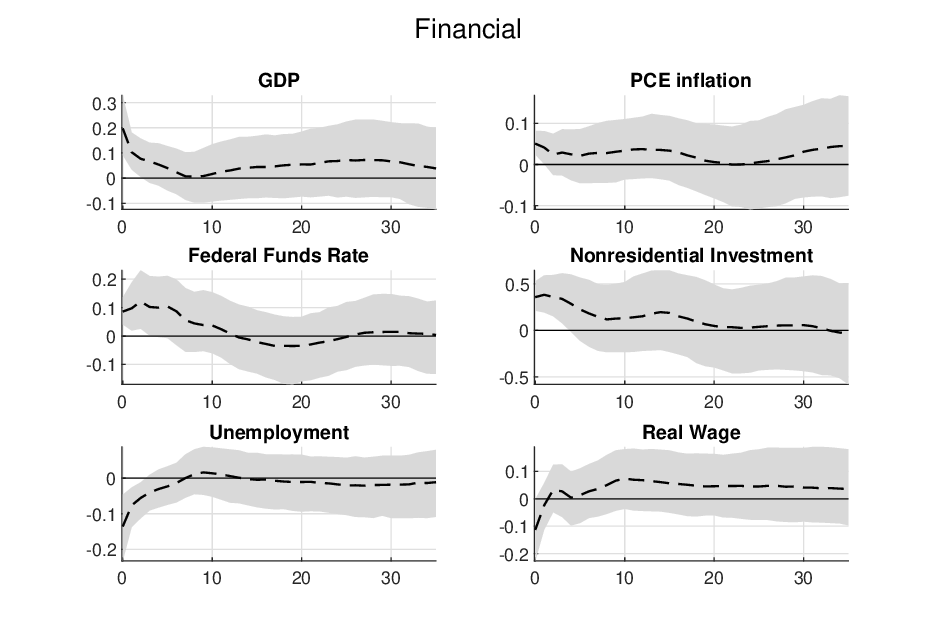}
   \caption{Impulse responses from a 35-variable VAR to a one-standard-deviation financial shock.}
   \label{fig:IRF_35_financial}	
\end{figure}

Compared to the results from the 15-variable VAR described in Section~\ref{ss:illustration}, responses of the overlapping variables in the 35-variable VAR tend to be smaller in magnitude, but are otherwise very similar. For example, Figure~\ref{fig:compare_35VARand15VAR_v1} reports the impulse responses of 6 selected variables from both VARs to a financial shock.\footnote{To obtain comparable results, we restrict the sample for the 35-variable VAR to the period from 1985Q1 to 2019Q4, so that it is the same as the setup in the 15-variable VAR.} Consistent with earlier results reported in \citet{FRS19} and \citet{chan22}, we find a relatively large effect on GDP and a smaller impact on prices; and both the responses of investment and stock prices are persistent. However, with more variables and sign restrictions to pin down the structural shocks, the credible bands from the 35-variable VAR tend to be narrower (more results are provided in Appendix B).

\begin{figure}[H]
	\begin{center}
		\includegraphics[height=8cm]{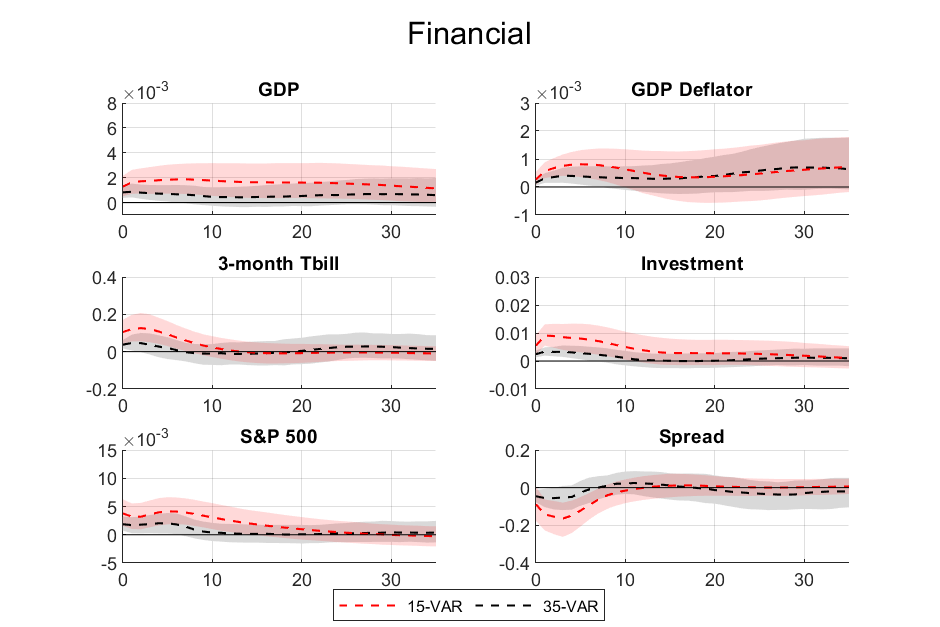}
	\end{center}
	\caption{Impulse responses of 6 selected variables to a one-standard-deviation financial shock from the 15- and the 35-variable VARs. %The shaded regions represent the 68\% credible intervals.
	}
	\label{fig:compare_35VARand15VAR_v1}
\end{figure}

Next, Figures \ref{fig:IRF_35_monetary} and \ref{fig:IRF_35_govt} plot the impulse responses of the same variables to the monetary policy shock and the  government spending shock. A contractionary monetary policy shock depresses output and inflation, while raising the unemployment rate and the real wage. In contrast, an expansionary government spending shock mostly raises inflation and short-term interest rate, and has negligible effects on output, unemployment or the real wage.

\begin{figure}[H]
    \centering
   \includegraphics[width=.8\textwidth]{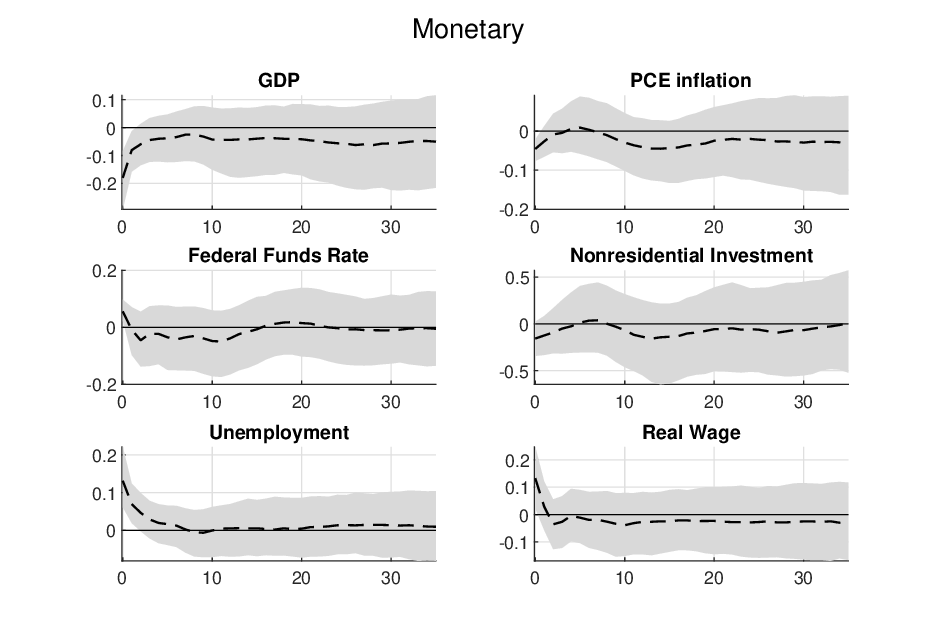}
   \caption{Impulse responses from a 35-variable VAR to a one-standard-deviation monetary policy shock.}
   \label{fig:IRF_35_monetary}	
\end{figure}

\begin{figure}[H]
    \centering
   \includegraphics[width=.8\textwidth]{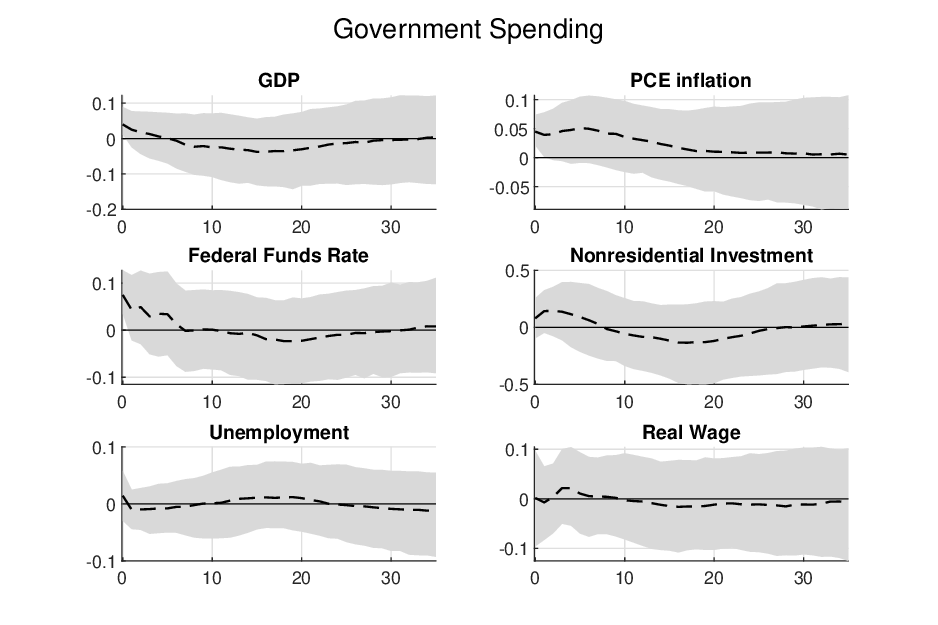}
   \caption{Impulse responses from a 35-variable VAR to a one-standard-deviation government spending shock.}
   \label{fig:IRF_35_govt}	
\end{figure}

\begin{figure}[H]
    \centering
   \includegraphics[width=.8\textwidth]{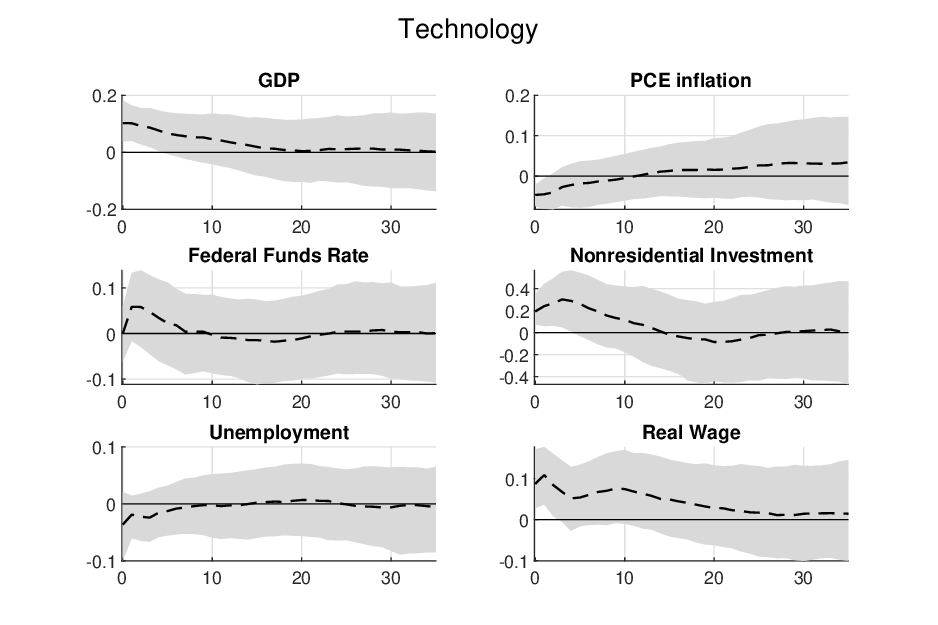}
   \caption{Impulse responses from a 35-variable VAR to a one-standard-deviation technology shock.}
   \label{fig:IRF_35_technology}	
\end{figure}

\begin{figure}[H]
    \centering
   \includegraphics[width=.8\textwidth]{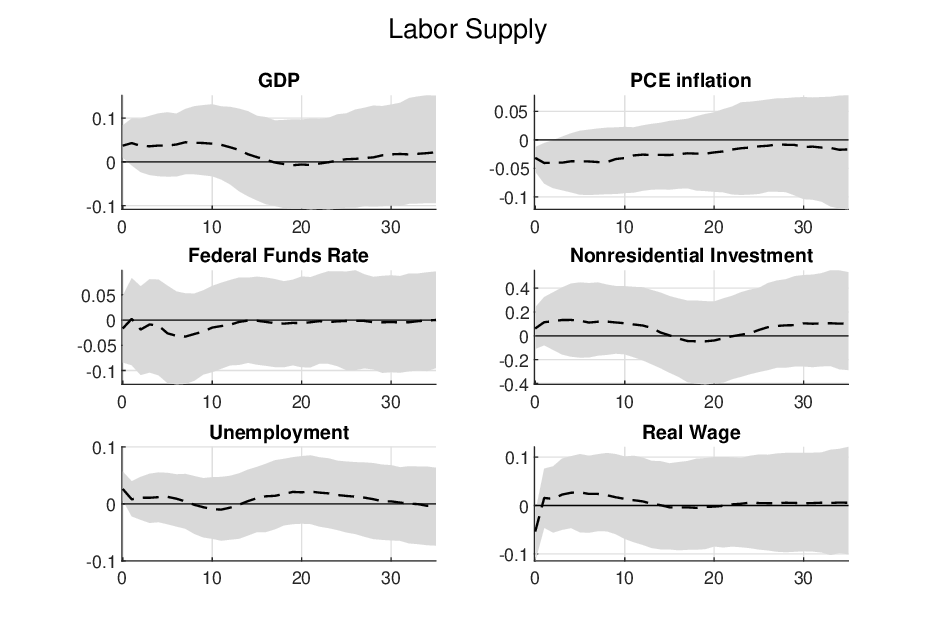}
   \caption{Impulse responses from a 35-variable VAR to a one-standard-deviation labor supply shock.}
   \label{fig:IRF_35_laborsupply}	
\end{figure}

\begin{figure}[H]
    \centering
   \includegraphics[width=.8\textwidth]{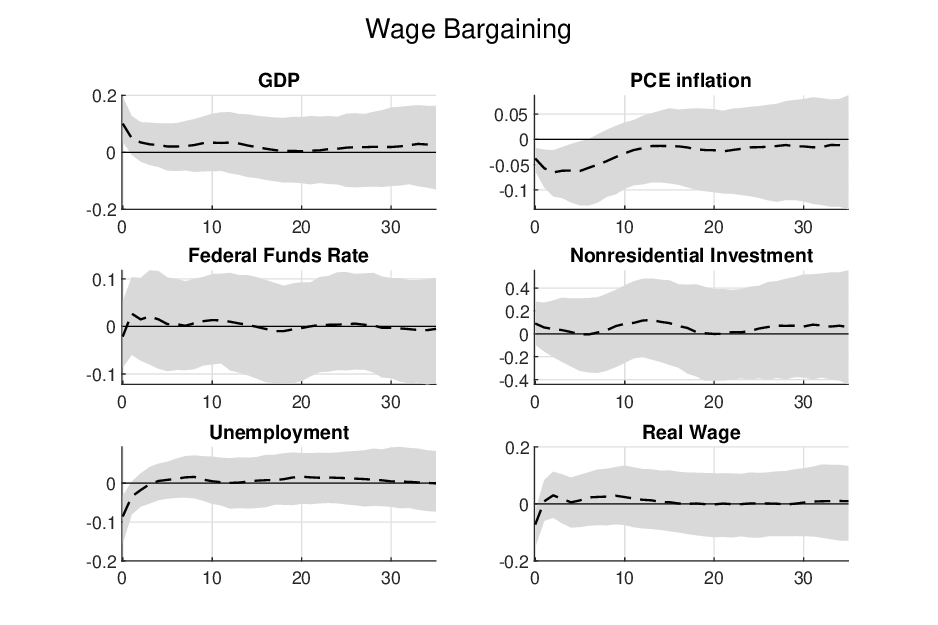}
   \caption{Impulse responses from a 35-variable VAR to a one-standard-deviation wage bargaining shock.}
   \label{fig:IRF_35_wagebargaining}	
\end{figure}

Finally, Figures \ref{fig:IRF_35_technology}--\ref{fig:IRF_35_wagebargaining}	report the impulse responses of the 6 variables to the 3 supply-type structural shocks: technology, labor supply and wage bargaining shocks. While all 3 supply-type structural shocks raise output and depress inflation, the technology shock has the largest impact on these two variables. In addition, the technology shock substantially increases real wage over a relatively long horizon, whereas the other two structural shocks have transient and negligible impacts on real wage.

Overall, this application demonstrates that it is practical to study the impacts of multiple structural shocks jointly in a large VAR. Using a large number of sign and ranking restrictions to identify different structural shocks, we are able to disentangle their differential effects on key macroeconomic variables.

\section{Concluding Remarks and Future Research} \label{s:conclusion}

Two recent developments have motivated our paper: the recognition of the need to include a large number of variables in structural analysis and the desire to use more credible structural restrictions to identify structural shocks. In response to these developments, we have introduced an efficient approach for estimating large VARs identified using a large number of sign and ranking restrictions on the impulse responses. We showed that the new approach is several orders of magnitude more efficient than the benchmark, reducing the computational time from days to seconds. We illustrated the methodology using a 35-variable VAR with sign and ranking restrictions to identify 8 structural shocks.

For future research, it would be useful to extend the proposed algorithms to impose both inequality and zero restrictions \citep{ARRW18}, where the latter may arise in proxy VARs \citep{CH19}. One possibility is to incorporate the Gibbs sampling approach developed in \citet{Read22} to draw the orthogonal matrix directly subject to zero and dynamic sign restrictions. Specifically, in higher dimensional settings, his algorithm could be used to draw the columns of $\bQ$ that are subject to zero or dynamic restrictions, whereas the proposed algorithm could then be applied to other columns that are not subject to zero or dynamic restrictions. It would also be interesting to incorporate richer prior information on the impact matrix or, more generally, impulse responses, as advocated in \citet{BH15} and \citet{BP23}.

\newpage 

\section*{Appendix A: Proof of Proposition}

In this appendix we provide a proof of the proposition stated in the main text.

\begin{proof}[Proof of Proposition \ref{prop:alg1}] Let $\bL$ denote the lower triangular Cholesky factor of $\vSigma$ such that $\vSigma = \bL\bL'$, and sample $\bQ$ uniformly from the orthogonal group $\bO(n)$, i.e., $\bQ \sim \distn{U}(\bO(n))$. Recall that $\mathcal{E}(\vSigma,\bQ)$ consists of all the permutations and sign switches of the columns of $\bL\bQ$. That is, an element $\bE \in \mathcal{E}(\vSigma,\bQ)$ can be represented as $\bE = \bL\bQ\bP\bD,$ where $\bP$ is an $n$-dimensional permutation matrix and $\bD$ is a diagonal matrix with elements $\pm 1$. Since the Haar measure is invariant under right multiplication of $\bP$ and $\bD$ \citep[see, e.g.,][Section 2.1.4]{muirhead}, $\bQ\bP\bD$ is a uniform draw from the orthogonal group $\bO(n)$.  Next, recall that $\mathcal{E}(\vSigma,\bQ,\mathcal{S}_0)$ denotes the (finite) subset of elements in $\mathcal{E}(\vSigma,\bQ)$ that satisfy all restrictions in $\mathcal{S}_0$. Step 4 of Algorithm~\ref{prop:alg1} uniformly obtains an element $\bR^*$ in $\mathcal{E}(\vSigma,\bQ,\mathcal{S}_0)$, which can be represented as $\bR^* = \bL\bQ\bP\bD$ for some permutation matrix $\bP$ and diagonal matrix $\bD$ with elements $\pm 1$. Hence, $\bQ^* = \bQ\bP\bD \sim \distn{U}(\bO(n))$ and $\bR^* = \bL\bQ^*$. Finally, since $\bR^*$ is an element in $\mathcal{E}(\vSigma,\bQ,\mathcal{S}_0)$, it satisfies all the restrictions in $\mathcal{S}_0$.
\end{proof}

\newpage 

\section*{Appendix B: Additional Empirical Results}

This appendix reports additional empirical results from the replication exercise of \citet{Uhlig05} and the empirical application. In particular, Figures~\ref{fig:Uhlig_proposed}-\ref{fig:Uhlig_Read}  report the dynamic responses of the 6 variables to a monetary policy shock identified using static and dynamic restrictions, compared using the proposed method and the algorithms of \citet{RWZ10} and \citet{Read22}. 

\begin{figure}[H]
	\begin{center}
		\includegraphics[height=8cm]{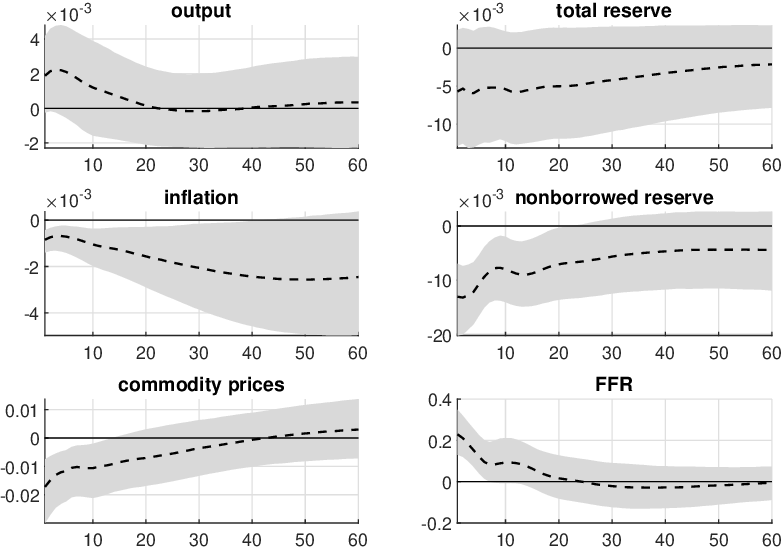}
	\end{center}
	\caption{Impulse responses to a monetary policy shock, using the proposed algorithm.}
	\label{fig:Uhlig_proposed}
\end{figure}

\begin{figure}[H]
	\begin{center}
		\includegraphics[height=8cm]{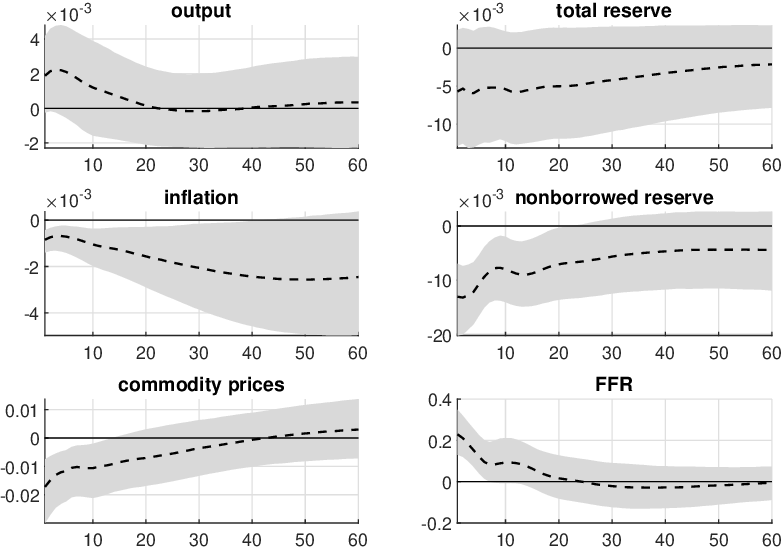}
	\end{center}
	\caption{Impulse responses to a monetary policy shock, using the algorithm of \citet{RWZ10}.}	
\end{figure}

\begin{figure}[H]
	\begin{center}
		\includegraphics[height=8cm]{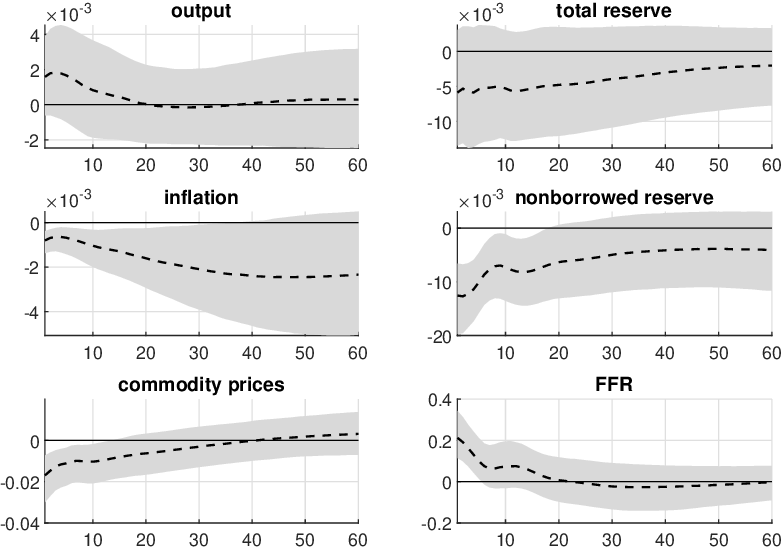}
	\end{center}
	\caption{Impulse responses to a monetary policy shock, using the algorithm of \citet{Read22}.}
	\label{fig:Uhlig_Read}
\end{figure}

Next, we report additional results from the empirical application that involves a 35-variable VAR. In particular, we compute the widths of 68\% credible bands of impulse responses to  a financial shock from the 35-variable VAR relative to those from the 15-variable VAR. The results are reported in Figure~\ref{fig:compare_35VARand15VAR_lengthratio}.
 
\begin{figure}[H]
	\begin{center}
		\includegraphics[height=8cm]{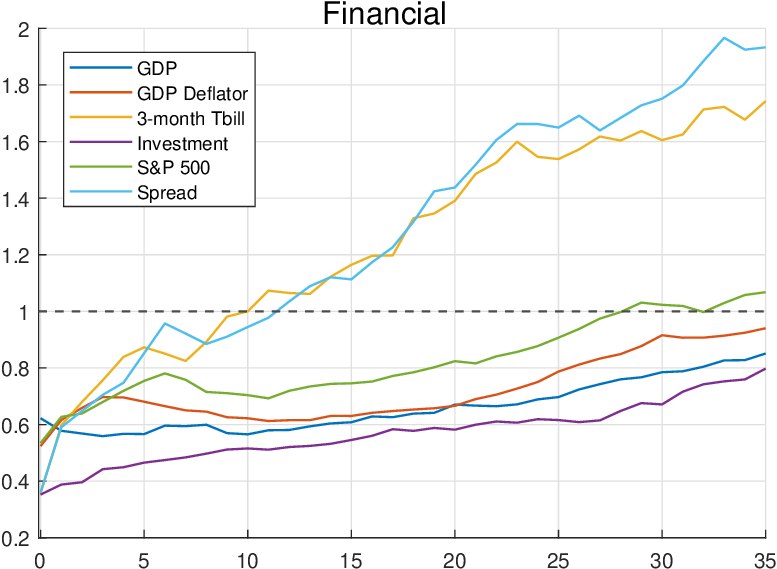}
	\end{center}
	\caption{Relative width of impulse-response credible bands between a 35-variable VAR and a 15-variable VAR under a financial shock. }
	\label{fig:compare_35VARand15VAR_lengthratio}
\end{figure}

It is clear from the plots that the credible bands from the 35-variable VAR tend to be narrower. For instance, for GDP, GDP deflator and investment, the credible bands from the 35-variable VAR are all smaller than those from the 15-variable VAR for all horizons. % For S\&P 500, the bands for 35-variable VAR are all smaller except a few large horizons. For 3-month Tbill and Spread, the bands are smaller before horizon 10. 

\newpage

\ifx\undefined\BySame
\newcommand{\BySame}{\leavevmode\rule[.5ex]{3em}{.5pt}\ }
\fi
\ifx\undefined\textsc
\newcommand{\textsc}[1]{{\sc #1}}
\newcommand{\emph}[1]{{\em #1\/}}
\let\tmpsmall\small
\renewcommand{\small}{\tmpsmall\sc}
\fi

\end{document}